\documentclass{daj}

\usepackage[utf8]{inputenc}
\usepackage{mathtools,amsmath}
\usepackage{amssymb}
\usepackage{microtype}
\usepackage{enumerate}
\usepackage{accents}
\usepackage{arydshln}

\usepackage{amsthm}
\usepackage{thmtools, thm-restate} 
\declaretheorem[name=Theorem, parent=section]{theorem}
\declaretheorem[name=Corollary, sibling=theorem]{corollary}

\declaretheorem[name=Lemma, sibling=theorem]{lemma}

\declaretheorem[name=Definition, sibling=theorem, style=definition]{definition}
\declaretheorem[name=Remark, sibling=theorem, style=definition]{remark}
\declaretheorem[name=Problem, sibling=theorem, style=definition]{problem}
\declaretheorem[name=Example, sibling=theorem, style=definition]{example}

\DeclareMathOperator{\supp}{supp}
\DeclareMathOperator{\subrank}{Q}
\DeclareMathOperator{\slicerank}{SR}
\DeclareMathAccent{\wtilde}{\mathord}{largesymbols}{"65}
\DeclareMathOperator{\asympsubrank}{\underaccent{\wtilde}{Q}}
\DeclareMathOperator{\asympslicerank}{\underaccent{\wtilde}{SR}}
\DeclareMathOperator{\asympsupslicerank}{{SR}\!\wtilde{\phantom{X}\!\!}}
\DeclareMathOperator{\asymprank}{\underaccent{\wtilde}{R}}
\DeclareMathOperator{\bordersubrank}{\underline{Q}}
\DeclareMathOperator{\tensorrank}{R}
\newcommand{\FF}{\mathbb{F}}

\newcommand{\RR}{\mathbb{R}}
\newcommand{\CC}{\mathbb{C}}
\newcommand{\NN}{\mathbb{N}}
\newcommand{\ZZ}{\mathbb{Z}}
\newcommand{\eps}{\varepsilon}
\DeclareMathOperator{\GL}{GL}

\DeclareMathOperator{\linspan}{span}
\DeclareMathOperator{\rank}{rank} \DeclareMathOperator{\col}{colspan}

\DeclareMathOperator{\maxrank}{max-rank}
\DeclareMathOperator{\minrank}{min-rank}
\DeclareMathOperator{\minrnk}{min-rank}
\DeclareMathOperator{\mincov}{min-cov}
\DeclareMathOperator{\minsupp}{min-supp}
\DeclareMathOperator{\maxsupp}{max-supp}

\newcommand{\id}{\mathrm{Id}}

  \newcommand{\beq}{\begin{equation}}
  \newcommand{\eeq}{\end{equation}}
  \newcommand{\beqn}{\begin{equation*}}
  \newcommand{\eeqn}{\end{equation*}}
  \newcommand{\beqr}{\begin{eqnarray}}
  \newcommand{\eeqr}{\end{eqnarray}}
  \newcommand{\beqrn}{\begin{eqnarray*}}
  \newcommand{\eeqrn}{\end{eqnarray*}}

\dajAUTHORdetails{%
  title = {Discreteness of Asymptotic Tensor Ranks}, 
  author = {Jop Briët, Matthias Christandl, Itai Leigh, Amir Shpilka, and Jeroen Zuiddam},
  plaintextauthor = {Jop Briët, Matthias Christandl, Itai Leigh, Amir Shpilka, Jeroen Zuiddam},
  keywords = {tensors, asymptotic rank, subrank, slice rank, discreteness},
}   

\dajEDITORdetails{%
   year={2025},
   number={15},
   received={27 September 2023},   
   revised={17 February 2025},    
   published={10 September 2025},  
   doi={10.19086/da.143834},       
}   

\begin{document}

\begin{frontmatter}[classification=text]

\title{Discreteness of Asymptotic Tensor Ranks} 

\author[briet]{Jop Briët\thanks{Supported by the Dutch Research Council (NWO) as part of the NETWORKS programme (Grant No.~024.002.003).}}
\author[christandl]{Matthias Christandl\thanks{Supported by the European Research Council (ERC Grant Agreement No.~81876), VILLUM FONDEN via the QMATH Centre of Excellence (Grant No.~10059) and the Novo Nordisk Foundation (Grant NNF20OC0059939 ``Quantum for Life'').}}
\author[leigh]{Itai Leigh\thanks{Supported by Erasmus+ International Credit Mobility grant, the Deutsch foundation, Israel Science Foundation (grant number 514/20) and the Len Blavatnik and the Blavatnik Family foundation.}}
\author[shpilka]{Amir Shpilka\thanks{Supported by Israel Science Foundation (grant number 514/20) and the Len Blavatnik and the Blavatnik Family foundation.}}
\author[zuiddam]{Jeroen Zuiddam\thanks{Supported by NWO Veni grant VI.Veni.212.284.}}

\begin{abstract}
    Tensor parameters that are amortized or regularized over large tensor powers, often called ``asymptotic'' tensor parameters, play a central role in several areas including algebraic complexity theory (constructing fast matrix multiplication algorithms), quantum information (entanglement cost and distillable entanglement), and additive combinatorics (bounds on cap sets, sunflower-free sets, etc.). Examples are the asymptotic tensor rank, asymptotic slice rank and asymptotic subrank.
    Recent works (Costa--Dalai, Blatter--Draisma--Rupniewski, Christandl--Gesmundo--Zuiddam) have investigated notions of discreteness (no accumulation points) or ``gaps'' in the values of such tensor parameters. 
    
    We prove a general discreteness theorem for asymptotic tensor parameters of order-three tensors and use this to prove that (1) over any finite field (and in fact any finite set of coefficients in any field), the asymptotic subrank and the asymptotic slice rank have no accumulation points, and (2) over the complex numbers, the asymptotic slice rank has no accumulation points.

    Central to our approach are two new general lower bounds on the asymptotic subrank of tensors, which measures how much a tensor can be diagonalized. The first lower bound says that the asymptotic subrank of any concise three-tensor is at least the cube root of the smallest dimension. The second lower bound says that any concise three-tensor that is ``narrow enough'' (has one dimension much smaller than the other two) has maximal asymptotic subrank.
%
\end{abstract}
\end{frontmatter}

\tableofcontents

\newpage
\section{Introduction}

Tensor parameters that are amortized or regularized over large tensor powers, often called ``asymptotic'' tensor parameters, play a central role in several areas of theoretical computer science including 
algebraic complexity theory (constructing fast matrix multiplication algorithms \cite{burgisser1997algebraic, blaser2013fast}, and barriers for such constructions \cite{alman2018limits,v017a002}), 
quantum information (entanglement cost and distillable entanglement~\cite{watrous_2018,brandao2016mathematics,vrana2015transformation}), 
and additive combinatorics (bounds on cap sets \cite{tao}, sunflower-free sets \cite{naslund_sawin_2017}, etc.). %
Examples are the asymptotic tensor rank (famous for its connection to the matrix multiplication exponent), the asymptotic subrank, and the asymptotic slice rank. %
These asymptotic tensor parameters are of the form $\underaccent{\wtilde}{F}(T) = \lim_{n\to \infty} F(T^{\boxtimes n})^{1/n}$ for some integer-valued function~$F$ (e.g.~tensor rank, subrank or slice rank), where $\boxtimes$ denotes the Kronecker product of tensors. 
The computation of these parameters, which in some applications is the main goal and in others is done to bound other parameters of interest, has turned out to be very challenging, and many questions about them are open despite much interest.

The fundamental question whether the matrix-multiplication exponent~$\omega$ equals~2 is closely related to the question whether asymptotic tensor rank is integral-valued. Contrary to matrix rank, some asymptotic tensor parameters may indeed take non-integral values.
For instance, the asymptotic slice rank of the $W$-tensor (which appears in the study of sunflower-free sets) equals $2^{h(1/3)} \approx 1.889$ \cite{strassen1991degeneration}, 
where $h$ is the binary entropy function, and the asymptotic slice rank of the cap set tensor (which appears in the study of arithmetic progression-free sets or cap sets) is known to be the non-integral value $\approx 2.755$ over the finite field $\FF_3$ \cite{MR3583358, tao, kleinberg2018growth, strassen1991degeneration}. 

This raises the fundamental question, for a given function $F$,
    what values $\underaccent{\wtilde}{F}(T)$ can take when varying $T$ over all tensors of order three with arbitrary dimensions (over any fixed field). 
    More generally, what is the structure (geometric, algebraic, topological, etc) of the set of values 
    \[
    \{\underaccent{\wtilde}{F}(T) : T \in \FF^{n_1} \otimes \FF^{n_2} \otimes \FF^{n_3}, n_1, n_2, n_3 \in \NN\}?
    \]
    Does $\underaccent{\wtilde}{F}(T)$ have accumulation points, that is, are there non-trivial sequences 
    of tensors $T_1, T_2, \ldots$ such that~$\underaccent{\wtilde}{F}(T_i)$ converges?
    Or is it discrete? 
    What ``gaps'' are there between the possible values? Even when $\FF$ is a finite field, the answers to these questions are a priori not clear at all.

In this paper we prove \emph{discreteness} of asymptotic tensor rank, asymptotic subrank and asymptotic slice rank in several regimes. This means that the values of each of these parameters have no accumulation points.
In fact, the proof of discreteness of asymptotic tensor rank (over any finite field or finite coefficient set in any field) follows from a simple argument. Using that same simple argument, together with several new results about tensors, we obtain discreteness for the other parameters. 
In particular, as a core ingredient, we prove a new result about diagonalizability of tensors. This comes in the form of a lower bound on the asymptotic subrank that relies only on the dimensions of the tensor (as opposed to the well-known laser method for fast matrix multiplication, which relies on much more information about the tensor). As another core ingredient we prove new results about matrix subspaces and their max-rank and min-rank. In particular, we prove that the max-ranks of the matrix spaces obtained by slicing a tensor in the three different directions are strongly related, in such a way that at least two of them must be large.

Our discreteness results show that there is a surprising rigidity in the asymptotic behaviour of tensors.
The discreteness of the above parameters gives rise to the phenomenon of ``rounding'' or ``boosting'' (upper or lower) bounds on them to the next possible value (although making this effective requires more knowledge of the possible values than just knowing discreteness). 
The discreteness of asymptotic tensor rank\footnote{Here we work in the regime where we fix the field to be any finite field, or we use any finite set of coefficients in any field.} implies, for instance, that the asymptotic rank of the matrix multiplication tensor is bounded away from (or equal to) the asymptotic tensor rank of any other tensor. In particular, the matrix multiplication exponent $\omega$ is an isolated point among the exponents of all bilinear maps.
If such a tensor is ``close enough'' to the matrix multiplication tensor, its exponent must ``snap'' to $\omega$. Similar statements hold for the other asymptotic parameters. In the context of combinatorial applications, this may moreover lead to limitations for the asymptotic slice rank to improve on existing results.

Before stating our results in detail, we will first discuss the various asymptotic ranks, their applications and context.


\subsubsection*{Matrix multiplication and asymptotic tensor rank}
Determining the computational complexity of matrix multiplication is a fundamental problem in algebraic complexity theory.  This complexity is controlled by the famous matrix multi\-pli\-cat\-ion exponent $\omega$, which is defined as the infimum over nonnegative real numbers~$\tau$ such that any two $n\times n$ matrices can be multiplied using $\mathcal{O}(n^{\tau})$ arithmetic operations \cite{burgisser1997algebraic,blaser2013fast}. The naive matrix multiplication algorithm gives the upper bound $\omega \leq 3$. In 1969, Strassen proved that $\omega \leq 2.81$  \cite{strassen1969gaussian}. Since then, using many different techniques, the best upper bound has been brought down to $\omega \leq 2.371552$ \cite{le2014powers, DBLP:conf/soda/AlmanW21, duan2023faster, williams2023new}. There is a tantalizing possibility (and many have conjectured) that $\omega = 2$, and routes have been proposed that aim to prove this \cite{cohn2003group, cohn2005group,cohn2013fast,MR3631613,blasiak2022matrix}. It is just as intriguing to consider the possibility that~$\omega > 2$ and $\omega$ giving rise to a new fundamental constant, and there has been much work on this lower bound direction as well \cite{DBLP:conf/stoc/BurgisserI11, MR3081636}.

Not only can we currently not determine the value of $\omega$, or decide whether $\omega = 2$ or~$\omega > 2$, there is a much more relaxed problem that we cannot solve. Indeed $\omega$ is naturally described in terms of tensors as the logarithm of the asymptotic tensor rank of the matrix multiplication tensor $\langle 2,2,2\rangle \in \FF^4 \otimes \FF^4 \otimes \FF^4$, that is $\omega = \log \asymprank(\langle 2,2,2\rangle)$, and thus $\omega > 2$ is equivalent to $\asymprank(\langle 2,2,2\rangle) > 4$. The following much more relaxed problem is open:

\begin{problem}[{\cite[Open Problem 15.5]{burgisser1997algebraic}}]\label{prob1}
    Prove that there is a tensor $T \in \FF^n \otimes \FF^n \otimes \FF^n$ with $\asymprank(T) > n$.
\end{problem}

It is possible that for every  
tensor $T \in \FF^n \otimes \FF^n \otimes \FF^n$ we have $\asymprank(T) \leq n$ (which would in particular imply $\omega=2$) and that the image of $\asymprank$ over all tensors is simply $\NN$. This naturally leads to the (very general) question:
    What is the structure (geometric, algebraic, topological, etc) of the set
    \[
    S = \{\asymprank(T) : T \in \FF^{n_1} \otimes \FF^{n_2} \otimes \FF^{n_3}, n_1, n_2, n_3 \in \NN\}?
    \]
Is there anything we can prove about $S$ without resolving \autoref{prob1} or determining~$\omega$? Not much is known. 

One known structural result is that $S$ is closed under applying any univariate polynomial with non-negative integer coefficients \cite[Theorem 4.8]{wigderson2022asymptotic}. This statement applies in fact more generally, and in particular also to asymptotic subrank and asymptotic slice rank. This thus says that $S$ has ``many'' elements. 

Our discreteness result says that $S$ does not have ``too many'' elements, and for asymptotic tensor rank over a finite field\footnote{Or in fact over any finite set of coefficients coming from an arbitrary field, see \autoref{th:asympran-disc}.} the proof is surprisingly simple, 
and uses some basic definitions that will play an important role throughout. Let~$\FF$ be any field. Let $T \in \FF^{n_1} \otimes \FF^{n_2} \otimes \FF^{n_3}$ be an order-three tensor over $\FF$ with dimensions $(n_1, n_2, n_3)$. 
The \emph{flattenings} of $T$ are the matrices in $(\FF^{n_1} \otimes \FF^{n_2}) \otimes \FF^{n_3}$, $\FF^{n_1} \otimes (\FF^{n_2} \otimes \FF^{n_3})$ and $\FF^{n_2} \otimes (\FF^{n_1} \otimes \FF^{n_3})$, obtained by naturally grouping the tensor legs of~$T$. 
We say two tensors $S \in \FF^{n_1} \otimes \FF^{n_2} \otimes \FF^{n_3}$ and $T \in \FF^{m_1} \otimes \FF^{m_2} \otimes \FF^{m_3}$ are \emph{equivalent} if there are linear maps $A_\ell : \FF^{n_\ell} \to \FF^{m_\ell}$ such that $(A_1 \otimes A_2 \otimes A_3) S = T$ and there are linear maps $B_\ell : \FF^{m_\ell} \to \FF^{n_\ell}$ such that $(B_1 \otimes B_2 \otimes B_3) S = T$.
A tensor $T \in \FF^{n_1} \otimes \FF^{n_2} \otimes \FF^{n_3}$ is called \emph{concise} 
if the three flattenings obtained by grouping two of the three tensor legs together, in the three possible ways, each have maximal rank. This means essentially that $T$ cannot be embedded into a smaller tensor space. Every tensor is equivalent to a concise tensor \cite[Section 14.3]{burgisser1997algebraic} (see also \autoref{lem:equiv-concise}).

Here is the proof that the asymptotic rank over finite fields is discrete:
Let $T_1, T_2, \ldots$ be any sequence of tensors such that $T_i \in \FF^{a_i} \otimes \FF^{b_i} \otimes \FF^{c_i}$ and such that $\{\asymprank(T_i) : i \in \NN\}$ is infinite. 
We may assume that every~$T_i$ is concise 
. Then $\asymprank(T_i) \geq \max \{a_i, b_i, c_i\}$. Since by assumption $\FF$ is finite, there are only finitely many tensors per format $a_i \times b_i \times c_i$, so the set of triples $\{(a_i, b_i, c_i) : i \in \NN\}$ is infinite. In particular, $\max \{a_i, b_i, c_i\}$ is unbounded, and so $\asymprank(T_i)$ is unbounded and cannot converge, which proves the claim. 

While this argument is very simple, a much more subtle argument and new technical results will be needed to deal with the other tensor parameters that we consider.

\subsubsection*{Asymptotic subrank and asymptotic slice rank}

Besides tensor rank, there are many other notions of rank of a tensor that play a role in applications, for instance the subrank, slice rank,  analytic rank~\cite{MR2773103,MR3964143,10.1215/00127094-2022-0086}, geometric rank~\cite{KopMosZui:GeomRankSubrankMaMu}, and G-stable rank~\cite{MR4471036}. We will focus here on the asymptotic subrank and asymptotic slice rank. The subrank was introduced by Strassen \cite{MR882307} in the study of matrix multiplication algorithms. The subrank $\subrank(T)$ is the size of the largest diagonal tensor that can be obtained from $T$ by taking linear combinations of the slices in the three different directions (i.e.~``Gaussian elimination'' for tensors). The slice rank was introduced by Tao \cite{tao} to give a tensor proof of the cap set problem after the first proof of Gijswijt and Ellenberg~\cite{MR3583358}. The slice rank $\slicerank(T)$ is the smallest number of tensors with flattening rank one whose sum is $T$. Tao proved that $\subrank(T) \leq \slicerank(T)$.
Recent works have shown that analytic rank, geometric rank \cite{DBLP:conf/stoc/CohenM21,10.1215/00127094-2022-0086}, and G-stable rank~\cite{MR4471036} are all equal to slice rank, up to a multiplicative constant. These results imply that the asymptotic versions of these parameters are all equal to the asymptotic slice rank, warranting our focus on it.

\emph{Slice rank method in combinatorics.} The proof of the longstanding cap set problem \cite{tao, sawin} (and other related results \cite{naslund_sawin_2017}) can be thought of as upper bounding the independence number of powers of a hypergraph, by constructing a tensor that ``fits'' on the hypergraph and then computing the slice rank of the powers of the tensor, that is, the asymptotic slice rank $\asympslicerank(T)$. Knowing (the structure of) the set
\[
S = \{\asympslicerank(T) : T \in \FF^{n_1} \otimes \FF^{n_2} \otimes \FF^{n_3}, n_1, n_2, n_3 \in \NN\}
\]
thus gives information on what bounds one can prove on the size of combinatorial objects using the slice rank method.
One of our main results is that asymptotic slice rank is discrete, not only over every finite field, but even over the complex numbers. The latter crucially requires a result from \cite{MR4495838} that characterizes asymptotic slice rank in terms of representation-theoretic objects called moment polytopes. In fact it is known by now that the four smallest values in $S$ are $0, 1, 1.889..., 2, 2.686...$ (see next section) and our result says that also the larger values are discrete.

\emph{Matrix multiplication barriers.} Besides the aforementioned combinatorial problems, the asymptotic subrank and asymptotic slice rank appear in several ``barrier results'' for matrix multiplication algorithms \cite{alman2018limits,v017a002,MR3631613}. Matrix multiplication algorithms are usually constructed by a reduction of matrix multiplication to another bilinear map, and these barrier results say what properties that intermediate bilinear map must have to obtain certain upper bounds on $\omega$, or to reach $\omega = 2$. These properties can be phrased in terms of asymptotic subrank or asymptotic slice rank. In particular, the barrier of \cite{v017a002} states that to reach $\omega = 2$, an intermediate tensor $T$ must have $\asymprank(T) = \asympsubrank(T)$, which has led to further research to find tensors with large asymptotic subrank \cite{DBLP:conf/mfcs/BlaserL20}. Our discreteness result says that asymptotic subrank is discrete over every finite field. We do not get this result over the complex numbers because the analogous representation-theoretic ingredient from above is missing here. Intriguingly, it is possible that 
\[
S = \{\asympsubrank(T) : T \in \FF^{n_1} \otimes \FF^{n_2} \otimes \FF^{n_3}, n_1, n_2, n_3 \in \NN\}
\]
equals the analogous set for asymptotic slice rank, that is, the following is open:
\begin{problem}
    Prove that asymptotic slice rank equals asymptotic subrank.
\end{problem}
We do as a side-result prove a new relation between asymptotic subrank and asymptotic slice rank (which we will discuss in more detail in the next section).

\subsubsection*{Previous work on discreteness of asymptotic ranks}
Several works, among which some very recent ones, have investigated notions of discreteness in the values of tensor parameters.

Strassen \cite[Lemma~3.7]{strassen1988asymptotic} proved that for any $k$-tensor $T$ over any field, the asymptotic subrank (and, as a consequence of his method, also the asymptotic slice rank) of~$T$ is equal to $0$, equal to $1$, or at least $2^{2/k}$. This result established the first ``gaps'' in asymptotic tensor parameters.
Costa and Dalai \cite{DBLP:journals/jcta/CostaD21} proved that, for any $k$-tensor $T$ over any field, the asymptotic slice rank of $T$ is equal to $0$, equal to $1$ or at least $2^{h(1/k)}$ where~$h$ is the binary entropy function\footnote{The binary entropy function is defined for $p \in (0,1)$ by $h(p) = - p\log_2 p - (1-p) \log_2 (1-p)$ and $h(0) = h(1) = 0$.}.
Christandl, Gesmundo and Zuiddam \cite{cgz} extended the result of Costa and Dalai by proving that, for any $k$-tensor $T$ over any field, the asymptotic subrank and asymptotic partition rank of $T$ are equal to~$0$, equal to~$1$ or at least~$2^{h(1/k)}$ (which is a tight bound). Additionally, they prove that for any $3$-tensor $T$ over any field, the asymptotic subrank and asymptotic slice rank of~$T$ are equal to~$0$, equal to $1$, equal to~$2^{h(1/3)}\approx 1.889$ or at least $2$. Gesmundo and Zuiddam~\cite{gesmundo2023gap} extended this result by proving that the next possible value after 2 is $\approx 2.686$.

Blatter, Draisma and Rupniewski \cite{bdr} proved that for any function on $k$-tensors, over any finite field, that is normalized and monotone the set of values that this function takes is well-ordered. The asymptotic (sub)rank and asymptotic slice rank\footnote{Over finite fields, since we do not know whether the limit $\lim_{n\to\infty}\slicerank(T^{\boxtimes n})^{1/n}$ exists in general, when we say asymptotic slice rank we will mean $\liminf_{n\to\infty} \slicerank(T^{\boxtimes n})^{1/n}$.} %
are examples of such functions. This means that the values of any such function do not have accumulation points ``from above'', but leaves open the possibility that there are accumulation points ``from below''.

Christandl, Vrana and Zuiddam \cite{MR4495838} proved that the asymptotic slice rank over the complex numbers takes only finitely many values on tensors of any fixed format, %
and thus only countably infinite many values in general. This is done by characterizing the asymptotic slice rank as an optimization over the moment polytope of the tensor and using the result that there are only finitely many such polytopes per tensor format.

Blatter, Draisma and Rupniewski \cite{https://doi.org/10.48550/arxiv.2212.12219} proved that for any ``algebraic'' tensor invariant over the complex numbers the related asymptotic parameter takes only countably many values. This implies in particular that the asymptotic (sub)rank, asymptotic slice rank, asymptotic geometric rank, and asymptotic partition rank (all over the complex numbers) take only countably many values. 

\subsubsection*{New results}

In this paper:
\begin{itemize}
    \item We prove two general lower bounds on the asymptotic subrank of concise 
    tensors that depend only on the dimensions of the tensor. The first says that the asymptotic subrank of any concise tensor is at least the cube root of its smallest dimension. The second says that the asymptotic subrank of any ``narrow enough'' tensor (meaning that one dimension is much smaller than the others) is maximal.
    \item We use those lower bounds to prove that over any finite set of coefficients in any field the asymptotic subrank has no accumulation points (i.e.,~is discrete). We moreover prove that over any finite set of coefficients and over the complex numbers the asymptotic slice rank has no accumulation points. A much simpler argument gives that the asymptotic rank is discrete over any finite set of coefficients. 
    \item As a core ingredient for the above, we prove optimal relations among the maximal rank of any matrix in the span of the slices of a tensor, when considering slicings in the three different directions.
    \item With similar techniques, we prove an upper bound on the difference between asymptotic subrank and asymptotic slice rank.
\end{itemize}

\subsection{Discreteness of asymptotic tensor parameters}

We will now discuss our results in more detail. We begin with some basic definitions (which we expand on in \autoref{subsec:prelim}).
Let $\FF$ be any field. Let $T \in \FF^{n_1} \otimes \FF^{n_2} \otimes \FF^{n_3}$ be an order-three tensor over $\FF$ with dimensions $(n_1, n_2, n_3)$. The subrank of $T$, denoted by $\subrank(T)$, is the largest number~$r$ such that there are linear maps $L_\ell : \FF^{n_\ell} \to \FF^r$ such that $(L_1 \otimes L_2\otimes L_3)T = \sum_{i=1}^r e_i \otimes e_i \otimes e_i$. In other words, the subrank measures how much a tensor can be diagonalized. 
The flattenings of $T$ are the elements in $(\FF^{n_1} \otimes \FF^{n_2}) \otimes \FF^{n_3}$, $\FF^{n_1} \otimes (\FF^{n_2} \otimes \FF^{n_3})$ and $\FF^{n_2} \otimes (\FF^{n_1} \otimes \FF^{n_3})$, obtained by naturally grouping the tensor legs of~$T$. 
We say $T$ has slice rank one if at least one of the flattenings has rank one. The slice rank of $T$, denoted by $\slicerank(T)$ is the smallest number $r$ such that~$T$ is a sum of $r$ tensors with slice rank one.
The asymptotic subrank of $T$ is defined as $\asympsubrank(T) = \lim_{n \to \infty} \subrank(T^{\boxtimes n})^{1/n}$ where $\boxtimes$ is the Kronecker product on tensors. The limit exists and equals the supremum by Fekete's lemma, since $\subrank$ is super-multiplicative. The asymptotic slice rank of $T$ we define as $\asympslicerank(T) = \liminf_{n\to\infty} \slicerank(T^{\boxtimes n})^{1/n}$. (Over the complex numbers, it is known that the liminf can be replaced by a limit. Over other fields, however this is generally not known.) 


%

%

\subsubsection{Discreteness of 
asymptotic slice rank and asymptotic subrank}
We prove discreteness (no accumulation points) for the asymptotic slice rank and asymptotic subrank, in several regimes, as follows.


\begin{theorem}\label{th:discr-finite}
    Over any finite set of coefficients in any field, the asymptotic subrank and the asymptotic slice rank each have no accumulation points.
\end{theorem}

\autoref{th:discr-finite} improves the result of Blatter, Draisma and Rupniewski~\cite{bdr} that the asymptotic subrank and asymptotic slice rank over any finite field have no accumulation points ``from above'', that is, have well-ordered sets of values. Indeed our result rules out all accumulation points (so also those ``from below''). We obtain the statement of \autoref{th:discr-finite}  also for asymptotic rank, but with a much shorter proof (\autoref{th:asympran-disc}).

\begin{theorem}\label{th:discr-complex}
    Over the complex numbers, the asymptotic slice rank has no accumulation points.
\end{theorem}

\autoref{th:discr-complex} improves the result of Blatter, Draisma and Rupniewski~\cite{https://doi.org/10.48550/arxiv.2212.12219} and Christandl, Vrana and Zuiddam \cite{MR4495838} that the asymptotic slice rank over the complex numbers takes only countably many values. Our result is indeed strictly stronger, as countable sets may a priori have accumulation points. %

Our results shed light on the recent results of Costa and Dalai \cite{DBLP:journals/jcta/CostaD21} and Christandl, Gesmundo and Zuiddam \cite{cgz} that found gaps between the smallest values of the asymptotic slice rank and asymptotic subrank, and answers positively (in some regimes) the question stated in \cite{cgz} asking whether the values will always be ``gapped''.

Besides the above, we prove discreteness theorems over arbitrary fields for two sub-classes of all tensors. Namely we consider the class of oblique tensors, which are the tensors whose support in some basis is an antichain, and the tight tensors, whose support in some basis can be characterized by algebraic equation (details in \autoref{sec:discreteness}). 
The set of tight tensors is a strict subset of the set of oblique tensors, which is a strict subset of all tensors. 
Both classes originate in the work of Strassen \cite{MR882307,strassen1988asymptotic,strassen1991degeneration}. Examples of tight tensors include the well-known matrix multiplication tensors. Tight tensors also play a central role in the laser method of Strassen to construct fast matrix multiplication algorithms. Our discreteness theorems in these regimes are as follows.

\begin{theorem}
    Over any field, on tight tensors,  the asymptotic subrank and asymptotic slice rank (which are equal) have no accumulation points.
\end{theorem}

\begin{theorem}
    Over any field, on oblique tensors, the asymptotic slice rank has no accumulation points.
\end{theorem}

\subsubsection{General discreteness theorem}

We prove the above discreteness theorems as an application of a general discreteness theorem that we discuss now. This general theorem gives discreteness for real-valued tensor parameters that satisfy several conditions. 

\begin{theorem}\label{gen-no-acc-simple}
    Let $\FF$ be any field. Let $\mathcal{C}$ be any subset of tensors over~$\FF$ such that for every $T \in \mathcal{C}$ there is an $S \in \mathcal{C}$ that is concise and equivalent to $T$.
    Let~$f : \mathcal{C} \to \RR_{\geq0}$ be any function such that
    \begin{enumerate}[\upshape (i)]
    \item $f$ does not change under equivalence of tensors
    \item\label{lb-cond-intro} $f(T) \geq \asympsubrank(T)$ for every $T \in \mathcal{C}$
    \item\label{fin-cond-intro} For every $n_1, n_2, n_3\in \NN$,  $f$ takes finitely many values on $\mathcal{C} \cap (\FF^{n_1} \otimes \FF^{n_2} \otimes \FF^{n_3})$ 
    \item For every $n_1, n_2, n_3\in \NN$, $f(T) \leq \min_\ell n_\ell$ for every $T \in \mathcal{C} \cap (\FF^{n_1} \otimes \FF^{n_2} \otimes \FF^{n_3})$.
    \end{enumerate}
    Then the set of values that $f$ takes on $\mathcal{C}$ has no accumulation points.
\end{theorem}

We make some remarks on condition~(\ref{fin-cond-intro}) of \autoref{gen-no-acc-simple}.
When applying \autoref{gen-no-acc-simple} to a given real-valued function $f$ on tensors, it will depend very much on the regime we are working in whether condition~(\ref{fin-cond-intro}) is trivial or not. In particular, over finite fields, condition~(\ref{fin-cond-intro}) is trivial, because then $\FF^{n_1} \otimes \FF^{n_2} \otimes \FF^{n_3}$ contains only finitely many elements. 
Another trivial situation is when~$f$ takes only integral values, as there are only finitely many integers between 0 and $\min_\ell n_\ell$ (in this case the conclusion of the theorem is also trivial). 
However, when $f$ is not integral (say $f$ is the asymptotic slice rank or asymptotic subrank) and the field is not finite (say $\FF$ is the complex numbers) condition~(\ref{fin-cond-intro}) can be very non-trivial. 
For instance, our application of \autoref{gen-no-acc-simple} to asymptotic slice rank over the complex numbers (leading to \autoref{th:discr-complex}) relies on the representation-theoretic characterization of this parameter in terms of moment polytopes \cite{MR4495838} and a result from invariant theory that there are only finitely many such polytopes per choice of~$(n_1, n_2, n_3)$.

Our proof of \autoref{gen-no-acc-simple} relies mainly on new lower bounds on the asymptotic subrank~$\asympsubrank(T)$ of concise tensors $T$, which we will discuss next. Intuitively, these lower bounds will ensure (using condition~(\ref{lb-cond-intro})) that for any infinite sequence of tensors $T_i$ the value of~$f(T_i)$ gets ``pushed up'' so much that it either cannot converge, or eventually becomes constant (when $\min_\ell n_\ell$ is bounded).

\subsection{Lower bounds on asymptotic subrank}

Having discussed our discreteness theorems, we will now discuss two results that are central in the proof of the discreteness theorems, and of independent interest. These results are about lower bounds on the asymptotic subrank. 

The general goal of these results, in the context of the proof of the general discreteness theorem, is to establish that if we have a sequence of tensors $T_i \in \FF^{a_i} \otimes \FF^{b_i} \otimes \FF^{c_i}$ for $i \in \NN$ such that the set of triples $\{(a_i, b_i, c_i) : i \in \NN\}$ is infinite, then the asymptotic subrank of these tensors must be either unbounded, or eventually constant and integral. We will discuss this in detail in the proof of the general discreteness theorem.



\subsubsection{Concise tensors have large asymptotic subrank}

For concise tensors~$T$ we prove a cube-root lower bound on the asymptotic subrank $\asympsubrank(T)$ in terms of the smallest dimension of the tensor.

\begin{theorem}\label{th:balanced}
Let $T \in \FF^{n_1} \otimes \FF^{n_2} \otimes \FF^{n_3}$ be concise. Then %
$\asympsubrank(T) \geq (\min_\ell n_\ell)^{1/3}$.
\end{theorem}

Strassen \cite{strassen1991degeneration}, building on the work of Coppersmith and Winograd \cite{MR1056627}, introduced a method to prove (optimal) asymptotic subrank lower bounds for a subclass of structured tensors called ``tight'' tensors. This method formed an integral part of the laser method for constructing fast matrix multiplication algorithms, and similar ideas have also been applied in the context of additive combinatorics \cite{kleinberg2018growth}.
We emphasize that \autoref{th:balanced} does not rely on any special structure of the tensor~$T$, unlike previous methods like the laser method that rely on $T$ being tight.

We do not know whether the lower bound $(\min_\ell n_\ell)^{1/3}$ in \autoref{th:balanced} is optimal. 
The best upper bound we know is from an example of a concise tensor $T \in \FF^n \otimes \FF^n \otimes \FF^n$ such that $\asympsubrank(T)=2\sqrt{n-1}$ (\autoref{ex:strassenNullAlgebra}). %

Alternatively, \autoref{th:balanced} can be phrased without the conciseness condition if we replace the dimension $n_\ell$ by the flattening rank $\tensorrank(T_\ell)$, as follows: Let $T$ be any tensor. Then $\asympsubrank(T) \geq (\min_\ell \tensorrank_\ell(T))^{1/3}$.

For symmetric tensors we can prove the following stronger bound. (In fact, we can prove this stronger bound for a larger class of tensors which we call ``pivot--matched'', as we will explain in \autoref{sub:sqrt-assub-equal-pivots}). %
We recall that a tensor $T=\sum_{i,j,k}T_{i,j,k}\, e_i\otimes e_j \otimes e_k \in \FF^{n} \otimes \FF^{n} \otimes \FF^{n}$ is called symmetric if permuting the three tensor legs does not change the tensor, that is, for every $(i_1,i_2,i_3)\in [n]^3$ and  every permutation  $\sigma\in S_3$ it holds that $T_{i_1,i_2,i_3}=T_{i_{\sigma(1)},i_{\sigma(2)},i_{\sigma(3)}}$. 


\begin{theorem}\label{th:balanced-sym}
Let $T \in \FF^{n} \otimes \FF^{n} \otimes \FF^{n}$ be concise and symmetric. Then $\asympsubrank(T) \geq n^{1/2}$.
\end{theorem}

We do not know whether the lower bound $n^{1/2}$ in \autoref{th:balanced-sym} is optimal.

Again, alternatively, \autoref{th:balanced-sym} can be phrased without conciseness if the dimension~$n$ is replaced by the flattening rank $\tensorrank_1(T)$ (for symmetric tensors the three flattening ranks are equal): Let $T \in \FF^n \otimes \FF^n \otimes \FF^n$ be symmetric. Then $\asympsubrank(T)\geq \tensorrank_1(T)^{1/2}$.

\subsubsection{Narrow tensors have maximal asymptotic subrank}

\autoref{th:balanced} implies that if $n_1$, $n_2$ and $n_3$ all grow, and $T \in \FF^{n_1} \otimes \FF^{n_2} \otimes \FF^{n_3}$ is any concise tensor, then $\asympsubrank(T)$ must also grow. This leaves open what happens in the regime where one of the $n_i$ is constant. We will consider the ``narrow'' regime where~$n_3=c$ is constant, and one of the dimensions $n_1, n_2$ is large enough. Here we prove that the asymptotic subrank is maximal, that is, matches the upper bound $c$.

\begin{theorem}\label{th:narrow-intro}
For every integer $c \geq 2$ there is an $N(c) \in \NN$ such that for every~$n_1, n_2$ with $\max\{n_1, n_2\} > N(c)$ and every concise tensor $T \in \FF^{n_1} \otimes \FF^{n_2} \otimes \FF^c$ we have $\asympsubrank(T) = c$.
\end{theorem}

Moreover, for the case $c = 2$ we prove with a direct construction that $N(2) = 2$ and that asymptotic subrank can be replaced by subrank.

\begin{theorem}
Let $n_1,n_2 > 2$. Let $T \in \FF^{n_1} \otimes \FF^{n_2} \otimes \FF^2$ be concise. Then $\subrank(T) = 2$.
\end{theorem}

\subsubsection{Lower bound on asymptotic subrank in terms of slice rank}

Besides the aforementioned bounds on the asymptotic subrank, we use some of the same methods to prove a lower bound on the asymptotic subrank in terms of the asymptotic slice rank.

Slice rank was introduced by Tao \cite{tao}. He proved that for every tensor $T$ we have $\slicerank(T) \geq \subrank(T)$. The gap between $\slicerank(T)$ and $\subrank(T)$ can be large, namely for generic tensors $T \in \FF^n \otimes \FF^n \otimes \FF^n$ (over algebraically closed fields $\FF$) it is known that $\slicerank(T) = n$ while $\subrank(T) = \Theta(\sqrt{n})$ \cite{derksen_et_al:LIPIcs.CCC.2022.9}. 
It is, however, not known whether there can be a large gap between $\slicerank(T^{\boxtimes m})$ and $\subrank(T^{\boxtimes m})$ when $m$ is large. In particular, it is possible that $\lim_{n\to\infty} \subrank(T^{\boxtimes n})^{1/n} = \lim_{n\to\infty} \slicerank(T^{\boxtimes n})^{1/n}$. Strassen's work implies this equality for the subset of tight tensors \cite{strassen1991degeneration}. Over the complex numbers, the limit $\lim_{n\to\infty} \slicerank(T^{\boxtimes n})^{1/n}$ is also known to exist and has a characterization in terms of moment polytopes \cite{MR4495838}.

We prove the following:

\begin{theorem}
\label{th:q-sr-bound-intro}
Suppose the limit $\asympslicerank(T) = \lim_{n\to\infty} \slicerank(T^{\boxtimes m})^{1/m}$ exists. 
Then 
\[
\asympsubrank(T) \geq \asympslicerank(T)^{2/3}.
\]
\end{theorem}

Our proof of \autoref{th:q-sr-bound-intro} consists of proving that the border subrank of the third power of $T$ is bounded from below in terms of the slice rank of $T$, and applying a field-agnostic Flanders-type lower bound on the max-rank of Haramaty and Shpilka~\cite{haramaty2010structure}.

\subsection{Max-rank and min-rank of slice spans}

Our lower bounds on asymptotic subrank discussed in the previous subsection rely on results we prove about the ranks of elements in the span of the slices of a tensor. These may be of independent interest and we discuss some of them here.

\subsubsection{Max-ranks of slice spans}

Our proof of \autoref{th:balanced} relies (among other ingredients) on the notion of the max-rank of matrix subspaces, and the relation between max-ranks of matrix subspaces obtained by slicing a tensor in the three possible directions.

For any matrix subspace $\mathcal{A} \subseteq \FF^{n_1} \otimes \FF^{n_2}$, we let $\textnormal{max-rank}(\mathcal{A})$ denote the largest matrix rank of any element of $\mathcal{A}$. 
To any tensor $T \in \FF^{n_1} \otimes \FF^{n_2} \otimes \FF^{n_3}$ we may associate three matrix subspaces 
$\mathcal{A}_1 \subseteq \FF^{n_2} \otimes \FF^{n_3}$, $\mathcal{A}_2 \subseteq \FF^{n_1} \otimes \FF^{n_3}$ and $\mathcal{A}_3 \subseteq \FF^{n_1} \otimes \FF^{n_2}$, obtained by taking the span of the slices of $T$ in each of the three possible directions. We denote the max-ranks of these spaces by $\subrank_\ell(T) = \textnormal{max-rank}(\mathcal{A}_\ell)$, for $\ell \in [3]$. We prove that for a concise tensor $T$ the max-ranks $\subrank_\ell(T)$ cannot all be small, in the following sense.

%
%

%


\begin{theorem}\label{th:uncertainty-principle-n-intro}
    Let $T \in \FF^{n_1} \otimes \FF^{n_2} \otimes \FF^{n_3}$ be concise. Let $\ell_1, \ell_2, \ell_3 \in [3]$ be distinct. Suppose that $|\FF|>n_{\ell_3}$. Then
    $\subrank_{\ell_1}(T)\subrank_{\ell_2}(T) \geq n_{\ell_3}.$
\end{theorem}

We have explicit examples of families of tensors (provided later) that show how \autoref{th:uncertainty-principle-n-intro} is essentially optimal:
\begin{itemize}
\item For every $n$ that is a square, 
there is a concise tensor $T \in \FF^{n} \otimes \FF^n \otimes \FF^n$ such that for all $\ell$ we have
$\sqrt{n} \leq \subrank_\ell(T)\leq 2\sqrt{n}$ (\autoref{ex:JopsEx}). In particular, for all $\ell_1 \neq \ell_2$ we have $n\leq\subrank_{\ell_1}(T)\subrank_{\ell_2}(T)\leq4n$. 

\item For every $n$, there is a concise tensor $T \in \FF^n \otimes \FF^n \otimes \FF^n$ such that $\subrank_1(T) = \subrank_3(T) = n$ and $\subrank_2(T) = 2$, so that $\subrank_2(T) \subrank_3(T) = 2n$ (\autoref{ex:strassenNullAlgebra}).

\item For every $c$ and for every $n$  that is a multiple of $c$,
there is a concise tensor $T \in \FF^{n} \otimes \FF^n \otimes \FF^n$ such that
$\subrank_1(T)=n$, $\subrank_2(T)\leq c+1$ and $\subrank_3(T)\leq n/c+1$ (\autoref{ex:generalisedNullAlg}).
In particular,
$\subrank_2(T)\subrank_3(T)\leq  \tfrac{(c+1)}{c} n+c+1$.
\end{itemize}

It follows from \autoref{th:uncertainty-principle-n-intro} (by a straightforward argument) that $\subrank_\ell(T)$ must be large for at least two directions $\ell$, in the following sense:
\begin{corollary}
    \label{th:rank-two-direc-intro}
    Let $T \in \FF^{n_1} \otimes \FF^{n_2} \otimes \FF^{n_3}$ be concise.
    Suppose that $|\FF|>\max_\ell n_\ell$.
    There are distinct $\ell_1, \ell_2 \in [3]$ such that $\subrank_{\ell_1}(T) \geq (\max_\ell n_\ell)^{1/2}$ and $\subrank_{\ell_2}(T) \geq (\min_\ell n_\ell)^{1/2}$.
\end{corollary}

From \autoref{th:rank-two-direc-intro}  we can prove a preliminary asymptotic subrank lower bound $\asympsubrank(T) \geq (\min_\ell n_\ell)^{1/4}$ using the (easy to prove) fact that $\asympsubrank(T)^2 \geq \subrank_{\ell_1}(T) \subrank_{\ell_2}(T)$ for any distinct ${\ell_1},{\ell_2} \in [3]$. Proving our stronger cube-root lower bound $\asympsubrank(T) \geq (\min_\ell n_\ell)^{1/3}$ of \autoref{th:balanced} requires slightly more work.

Work on max-rank (and its relation to dimension) goes back to Dieudonn\'{e} \cite{MR29360}, Meshulam \cite{meshulam1985maximal} and \cite{flanders1962spaces}. Another relevant line of work here is on commutative and non-commutative rank, which has established strong connections between max-rank and non-commutative rank \cite{cohn1995skewfields,MR2123060,derksen2018ncrk}.

\subsubsection{Min-ranks of slice spans}

Our proof of \autoref{th:narrow-intro} relies on a careful analysis of the min-rank of matrix subspaces, the relation to subrank and the behaviour of min-rank under powering.

For any matrix subspace $\mathcal{A} \subseteq \FF^{n_1} \otimes \FF^{n_2}$, we let $\minrank(\mathcal{A})$ denote the smallest matrix rank of any nonzero element of $\mathcal{A}$. We prove several properties of the min-rank, of which we give a rough outline here (the precise description we defer to later):
\begin{itemize}
    \item If the slices of a tensor have large min-rank, then the tensor has large subrank (\autoref{cla-high-rank-high-subrank}). %
    \item Any concise tensor has a slice of large rank (\autoref{cla:high-rk}).
    \item If a matrix subspace has large max-rank, then we can transform the subspace in a natural fashion such that it has large min-rank and all elements are diagonal (\autoref{claim:matrices-into-diagonal-with-large-rank}).
    \item Min-rank is super-multiplicative under tensor product, as long as at least one of the matrix subspaces is diagonal (\autoref{lem:minrank-diag-supermultiplicative}).
\end{itemize}
A careful combination of the above ingredients leads to a proof of \autoref{th:narrow-intro}.

The min-rank has been investigated before in several different contexts. Amitsur~\cite{amitsur1965generalized} used min-rank to characterize properties of rings of operators. %
Meshulam and Semrl \cite{MR2073900} used the min-rank to study properties of operator spaces. %
In quantum information the rank is a measure of entanglement. Spaces of bipartite states which are all entangled are spaces of matrices over the complex numbers with min-rank strictly greater than~$1$. They were investigated in \cite{wallach2002unentangled} and \cite{parthasarathey2004maxdim}. In \cite{hayden2006generic} a slightly different angle was taken, analysing random subspaces. It was shown that most random subspaces have almost maximal min-rank, which was used for superdense coding in \cite{anura2006superdense}, and was summarised in \cite{hayden2004random}.
 Generalising both of these lines of work,
in \cite{cubitt2008subspacesbounded} the question of dimension versus min-rank was addressed as number of qubits versus guaranteed entanglement in a subspace of states, and their construction is used in a follow-up paper \cite{cubitt2008counterexamples} to show that 
there are counterexamples to the additivity of $p$-R{\'e}nyi entropies, for all $p\leq p_0$ for some small constant $p_0<1$, utilising the fact that the $0$-R{\'e}nyi entropy is the min-rank.

\subsection{Tensor preliminaries}\label{subsec:prelim}

We summarize in this section some standard tensor preliminaries and terminology that we will use throughout the paper, following the papers of Strassen \cite{MR882307,strassen1988asymptotic,strassen1991degeneration} and the book of Bürgisser, Clausen and Shokrollahi~\cite{burgisser1997algebraic}.

\subsubsection*{Tensors}
Let $\FF$ be any field. For any $n \in \NN$ we denote by $e_1, \ldots, e_n$ the standard basis elements of the vector space $\FF^n$. For any $n_1, n_2, n_3 \in \NN$, we denote by $\FF^{n_1} \otimes \FF^{n_2} \otimes \FF^{n_3}$ the vector space of order-three tensors $T = \sum_{i \in [n_1]} \sum_{j \in [n_2]} \sum_{k \in [n_3]} T_{i,j,k}\, e_i \otimes e_j \otimes e_k$ with coefficients $T_{i,j,k} \in \FF$. In most of this paper we will be dealing only with order-three tensors, so instead of order-three tensor we will just say tensor. 

\subsubsection*{Restriction and equivalence} 
We will be using a natural preordering on tensors called restriction, which is defined as follows. For any two tensors $T \in \FF^{n_1} \otimes \FF^{n_2} \otimes \FF^{n_3}$ and $S \in \FF^{m_1} \otimes \FF^{m_2} \otimes \FF^{m_3}$, we say $T$ \emph{restricts to} $S$ and write $T \geq S$ if there are linear maps $L_\ell : \FF^{n_\ell} \to \FF^{m_\ell}$ such that $(L_1 \otimes L_2 \otimes L_3)T = S$. We say $S$ and $T$ are \emph{equivalent} if $T \geq S$ and $S \geq T$. Most tensor properties we will study are invariant under equivalence and monotone under restriction (see \autoref{lem:monotone-restrict}). Finally, we say $S, T \in \FF^{n_1} \otimes \FF^{n_2} \otimes \FF^{n_3}$ are \emph{isomorphic} if $(L_1 \otimes L_2 \otimes L_3)T = S$ for invertible linear maps $L_\ell : \FF^{n_\ell} \to \FF^{n_\ell}$.

\subsubsection*{Flattenings and flattening ranks}
Let $T \in V_1 \otimes V_2 \otimes V_3$. Grouping together $V_2$ and $V_3$ gives an element  $T_{1} \in V_1 \otimes (V_2 \otimes V_3)$, and we similarly obtain $T_2$ and $T_3$. 
We may think of the $T_\ell$ as matrices and we call them 
the \emph{flattenings} of $T$.
Let $\tensorrank_{\ell}(T)$ ve the matrix rank of $T_{\ell}$. We call $\tensorrank_1(T)$, $\tensorrank_2(T)$ and $\tensorrank_3(T)$ 
the \emph{flattening ranks} of $T$.
The triple $(\tensorrank_1(T), \tensorrank_2(T), \tensorrank_3(T))$ is often called the \emph{multilinear rank} of $T$.

\subsubsection*{Conciseness} 
Let $T \in \FF^{n_1} \otimes \FF^{n_2} \otimes \FF^{n_3}$. From the definition of flattening rank it follows that $\tensorrank_\ell(T) \leq n_\ell$ for every $\ell \in [3]$. We call $T$ \emph{concise} if $\tensorrank_\ell(T) = n_\ell$  for every~$\ell \in [3]$. Conciseness essentially says that the tensor cannot fit in a smaller tensor space. The property of being concise is very mild in the following sense. We call a tensor $S$ a \emph{subtensor} of $T = \sum_{i,j,k} T_{i,j,k}\, e_i \otimes e_j \otimes e_k$ if $S$ is of the form $S = \sum_{i\in I,j\in J,k\in K} T_{i,j,k}\, e_i \otimes e_j \otimes e_k$ for some subsets $I,J,K$.

\begin{lemma}\label{lem:equiv-concise}
    Every tensor is equivalent to a concise tensor, which is unique up to equivalence. Concretely, every tensor is equivalent to a subtensor that is concise.
\end{lemma}
\begin{proof}
    The first statement is a standard fact \cite[Section 14.3]{burgisser1997algebraic}. For the second statement, let $T = \sum_{i,j,k} T_{i,j,k}\, e_i \otimes e_j \otimes e_k \in \FF^{n_1} \otimes \FF^{n_2} \otimes \FF^{n_3}$. Let $A_i = (T_{i,j,k})_{j,k} \in \FF^{n_2} \otimes \FF^{n_3}$ for $i \in [n_1]$ be the 1-slices of~$T$. Then one verifies that the dimension of the span of the~$A_i$ equals the flattening rank $r_1=\tensorrank_1(T)$. Suppose $A_1, \ldots, A_r$ are linearly independent, and let $S$ be the tensor with these matrices as 1-slices. Then $S$ is a subtensor of $T$ so clearly $T \geq S$. Also, since the matrices $A_{r+1}, \ldots, A_{n_1}$ are in the span of the matrices $A_1, \ldots, A_r$ we have that $S \geq T$. Since $S$ and $T$ are thus equivalent, their flattening ranks are equal. Repeating this construction for the other two directions gives the claimed subtensor.
    %
\end{proof}

Not every space $\FF^{n_1} \otimes \FF^{n_2} \otimes \FF^{n_3}$ contains concise tensors, as can be seen from the following implication (which is in fact known to be an equivalence \cite[Theorem~2]{MR2776439}):

\begin{lemma}\label{cl:concise-formats}
For any $n_1,n_2,n_3\in\NN$, if there is a concise tensor in $\FF^{n_1} \otimes \FF^{n_2} \otimes \FF^{n_3}$, then $n_1 \leq n_2\cdot n_3$, $n_2 \leq n_1\cdot n_3$, and $n_3 \leq n_1\cdot n_2$.
\end{lemma}
\begin{proof}
    Suppose $T \in \FF^{n_1} \otimes \FF^{n_2} \otimes \FF^{n_3}$ is concise. Then $n_{\ell_1} = \rank(T_{\ell_1})$ where $T_{\ell_1}$ is the previously defined flattening of $T$. Since $T_{\ell_1}$ is an $n_{\ell_1} \times (n_{\ell_2}\cdot n_{\ell_3})$ matrix, we have $\tensorrank(T_{\ell_1}) \leq n_{\ell_2} \cdot n_{\ell_3}$. Thus $n_{\ell_1} \leq n_{\ell_2} \cdot n_{\ell_3}$ for every choice of distinct ${\ell_1},{\ell_2},{\ell_3} \in [3]$.
\end{proof}

\subsubsection*{Unit tensors, tensor rank, subrank, and slice rank}
For every $r \in \NN$ we define the tensor $\langle r \rangle = \sum_{i=1}^r e_i \otimes e_i \otimes e_i \in \FF^r \otimes \FF^r \otimes \FF^r$. We call~$\langle r \rangle$ the \emph{unit tensor} of rank $r$.
We say a tensor has \emph{tensor rank one} if it is of the form $u \otimes v \otimes w$ for nonzero vectors $u,v,w$.
The tensor rank of a tensor $T$ is the smallest number $r$ such that $T$ can be written as a sum of $r$ tensors with tensor rank one. 
We denote tensor rank by $\tensorrank(T)$. Equivalently, the tensor rank $\tensorrank(T)$ is the smallest number $r$ such that $T \leq \langle r\rangle$. The \emph{subrank} of $T$ is the largest number $r$ such that $\langle r\rangle \leq T$. We denote subrank by $\subrank(T)$. 
We say a tensor has \emph{slice rank one} if any of its flattening ranks equals~one. The slice rank of a tensor $T$ is the smallest number $r$ such that $T$ can be written as a sum of $r$ tensors with slice rank one.
Subrank, slice rank, the flattening ranks and tensor rank satisfy $\subrank(T) \leq \slicerank(T) \leq \tensorrank_\ell(T) \leq \tensorrank(T)$ for every $\ell\in [3]$.

It is also easy to see that these measures are monotone under restrictions:

\begin{lemma}\label{lem:monotone-restrict}
    For any two tensors $T$ and $S$, if $T\geq S$, then we have $\tensorrank(T)\geq \tensorrank(S)$, $\slicerank(T)\geq \slicerank(S)$, $\subrank(T)\geq \subrank(S)$, and for every $\ell\in[3]$, 
    $\tensorrank_\ell(T)\geq \tensorrank_\ell(S)$.
\end{lemma}

\subsubsection*{Kronecker product and asymptotic tensor parameters}
Asymptotic tensor parameters are defined by ``regularizing'' a tensor parameter over large powers under the Kronecker product. For any two tensors 
\begin{align*}
T &= \sum_{i,j,k} T_{i,j,k}\, e_i \otimes e_j \otimes e_k \in \FF^{n_1} \otimes \FF^{n_2} \otimes \FF^{n_3}\\
S &= \sum_{a,b,c} S_{a,b,c} \, e_a\otimes e_b \otimes e_c \in \FF^{m_1} \otimes \FF^{m_2} \otimes \FF^{m_3}
\end{align*}
the Kronecker product $T \boxtimes S \in \FF^{n_1 m_1} \otimes \FF^{n_2 m_2} \otimes \FF^{n_3 m_3}$ is defined as
\[
T \boxtimes S = \sum_{i,j,k} \sum_{a,b,c} T_{i,j,k} S_{a,b,c}\, (e_i \otimes e_{a}) \otimes (e_j \otimes e_{b}) \otimes (e_k \otimes e_{c}).
\]
We define $T^{\boxtimes n}$ as the product of $n$ copies of $T$. The asymptotic tensor rank is defined as $\asymprank(T) = \lim_{n \to \infty} \tensorrank(T^{\boxtimes n})^{1/n}$, which, because of sub-multiplicativity of $\tensorrank$, equals the infimum $\inf_{n} \tensorrank(T^{\boxtimes n})^{1/n}$ (by Fekete's lemma). The asymptotic subrank is defined as $\asympsubrank(T) = \lim_{n \to \infty} \subrank(T^{\boxtimes n})^{1/n}$, which, because of super-multiplicativity of the subrank $\subrank$, equals the supremum $\sup_{n} \subrank(T^{\boxtimes n})^{1/n}$. Over the complex numbers, the asymptotic slice rank is defined as $\asympslicerank(T) = \lim_{n\to\infty} \slicerank(T^{\boxtimes n})^{1/n}$. By the characterization in \cite{MR4495838} this limit exists. Over other fields, we generally do not know whether the limit $\lim_{n\to\infty} \slicerank(T^{\boxtimes n})^{1/n}$ exists, except for oblique tensors. When we do not know whether the limit exists we will instead consider the liminf and refer to this as the ``asymptotic slice rank''.\footnote{The results we have over finite fields for asymptotic slice rank hold equally well using liminf or limsup.}

\section{Max-rank of slice spans and asymptotic subrank lower bound}\label{sec:max-rank}

In this section we will prove a lower bound on the asymptotic subrank.
The main goal is to obtain a lower bound on the asymptotic subrank of concise tensors that only depends on the dimensions of the tensor and that moreover grows as the dimensions grow. Unlike previous lower bound methods (like the laser method mentioned in the introduction) our lower bound does not rely on any special structure of the tensor. Indeed the property of being concise is very mild as any tensor can be made concise by embedding it in the smallest tensor space that it fits in.  %

Concretely we will prove that the asymptotic subrank of a concise tensor is at least the cube root of the smallest dimension of the tensor. We will give two proofs of this in this paper, one in this section and another one in \autoref{subsec:cube-root-via-pivots} (which uses different ideas which we then use to get stronger results for symmetric tensors).

The proof of the asymptotic subrank lower bound in this section relies on the notion of max-rank of a matrix subspace (\autoref{subsec:maxrank}). 
We prove that there is a strong relation between the max-ranks of slice spans of any tensor. Namely, the product of every pair of them is large (\autoref{subsec:QiQj}). This result, combined with known facts about matrix multiplication tensors, leads to the asymptotic subrank bound (\autoref{subsec:lb-mamu}).

\subsection{Max-rank of matrix subspaces and slice spans of a tensor}\label{subsec:maxrank}

    We denote by $\FF^{n_1} \otimes \FF^{n_2}$ the space of $n_1 \times n_2$ matrices over $\FF$.
    We define the \emph{max-rank} of any matrix subspace $\mathcal{A} \subseteq \FF^{n_1} \otimes \FF^{n_2}$ as
    \[
\maxrank(\mathcal{A}):=\max\{\rank(A):A\in\mathcal{A}\}.
    \]
For any tensor $T = \sum_{i,j,k} T_{i,j,k}\, e_i \otimes e_j \otimes e_k \in \FF^{n_1} \otimes \FF^{n_2} \otimes \FF^{n_3}$ we define the matrices %
\begin{align*}
T^{(1)}_i &= \sum_{j,k} T_{i,j,k}\, e_j \otimes e_k \quad (i \in [n_1])\\
T^{(2)}_j &= \sum_{i,k} T_{i,j,k}\, e_i \otimes e_k \quad (j \in [n_2])\\
T^{(3)}_k &= \sum_{i,j} T_{i,j,k}\, e_i \otimes e_j \quad (k \in [n_3]).
\end{align*}
For every $\ell \in [3]$, we call $\{T^{(\ell)}_i : i \in [n_\ell] \}$ the \emph{$\ell$-slices} of $T$. %
We define the $\ell$-slice span of $T$ as the matrix subspace
$\mathcal{A}_\ell = \linspan \{T^{(\ell)}_i : i \in [n_\ell] \}$
and denote its max-rank by
\[
\subrank_\ell(T) = \maxrank(\mathcal{A}_\ell).
\]
We note that the flattening rank $\tensorrank_\ell(T)$ defined earlier, is equal to the dimension of the subspace $\mathcal{A}_\ell$.

\subsection{Lower bound on products of max-ranks of slice spans}\label{subsec:QiQj}

%

%

%

%

In this section we prove the aforementioned fundamental property of the max-ranks of the slice spans of a tensor, showing that they cannot all be small. This will be a crucial ingredient for proving the asymptotic subrank lower bound afterwards.

\begin{theorem}\label{th:uncertainty-principle-n}
    Let $T \in \FF^{n_1} \otimes \FF^{n_2} \otimes \FF^{n_3}$ be concise. Let ${\ell_1}, {\ell_2}, {\ell_3} \in [3]$ be distinct. Suppose that $|\FF|>n_{\ell_3}$. Then 
    \[
    \subrank_{\ell_1}(T)\subrank_{\ell_2}(T) \geq n_{\ell_3}.
    \]
\end{theorem}

A corollary of this is that there are two distinct directions with large max-rank:

\begin{corollary} %
\label{cor:twoLargeRankDirs}
    Let $T\in\FF^{n_1}\otimes\FF^{n_2}\otimes\FF^{n_3}$ be concise.
    Suppose $n_1 \leq n_2 \leq n_3$. Suppose that $|\FF| > n_3$. 
    Then there are distinct ${\ell_1},{\ell_2}\in[3]$ such that
    \begin{align*}\subrank_{\ell_1}(T) &\geq \sqrt{n_3},\\
    \subrank_{\ell_2}(T) &\geq \sqrt{n_1}.
    \end{align*}
\end{corollary}
\begin{proof}%
    We have by \autoref{th:uncertainty-principle-n} that
    \begin{align*}
    \subrank_1(T) \subrank_2(T) &\geq n_3\\
    \subrank_2(T) \subrank_3(T) &\geq n_1\\
    \subrank_1(T) \subrank_3(T) &\geq n_2 \geq n_1.
    \end{align*}
    Thus one of $\subrank_1(T)$ and $\subrank_2(T)$ is at least $\sqrt{n_3}$, and either $\subrank_3(T)$ is at least~$\sqrt{n_1}$ or both $\subrank_1(T)$ and $\subrank_2(T)$ are at least $\sqrt{n_1}$.
\end{proof}

We will now work towards the proof of \autoref{th:uncertainty-principle-n}.

For matrices $A \in \FF^{n\times m}$ and $B \in \FF^{n \times k}$ we denote by $[A; B] \in \FF^{n \times (m+k)}$ the concatenation of $A$ and $B$. For any matrix $A$ we denote by $A|_{[d]}$ the matrix consisting of the first $d$ columns of $A$. We denote by $\col(A)$ the span of the columns of $A$.

\begin{lemma}\label{lem:colspan}
    Let $A \in \FF^{n\times m}$ and $B \in \FF^{n \times k}$. There is a nonzero 
    polynomial $p\in \FF[x_{1,1},x_{1,2},,\ldots,x_{i,j},\dots,x_{k,k}]$ of degree $s = \rank([A;B]) - \rank(A)$ such that for any matrix $U \in \FF^{k\times k}$ satisfying $p(U)\neq 0$, we have that
    \[
    \col([A;B]) = \col([A; (BU)|_{[s]}]).
    \]
\end{lemma}
\begin{proof}
    Let $r = \rank([A;B])$. Then, there is an $r\times r$ submatrix of $[A;B]$ with rank~$r$ induced by $r-s$ columns of~$A$ and~$s$ columns of~$B$.
    Let~$R\subseteq [n]$ be the set of row indices of this submatrix, and let $C_A\subseteq[m], C_B\subseteq \{m+1,\dots,m+k\}$ be the respective column indices of~$A$ and~$B$ in this submatrix.
    Let $p(X)$ be the determinant of the submatrix of $[A;BX]$ induced by the rows in~$R$ and columns in $C_A\cup \{m+1,\dots,m+s\}$, where $X = (x_{i,j})_{i,j=1}^k$ is a matrix of variables.
    
    We claim that~$p$ is nonzero. There is a permutation matrix~$V\in \FF^{k \times k}$ that swaps the columns in~$C_B$ with the columns in $\{m+1,\dots, m+s\}$. Since the submatrix of $[A;BV]$ induced by the rows in~$R$ and columns in $C_A\cup \{m+1,\dots,m+s\}$ is then of full rank, it follows that~$p(V)\ne 0$. 
        
    Let $U \in \FF^{k \times k}$ be a matrix such that $p(U) \neq 0$. 
    Then we have $\rank([A; (BU)|_{[s]}]) \geq r$. From $\col([A; (BU)|_{[s]}])\subseteq \col([A;B])$ and $\rank([A;B]) = r$, we find the equality $\col([A;B]) = \col([A; (BU)|_{[s]}])$.
\end{proof}




\begin{lemma}\label{lem:max-rank-induction}
    Let $A_1, \ldots, A_m \in \FF^{n \times k}$ be matrices. Suppose $|\FF| > \rank([A_1; \cdots; A_m])$. Let 
    \[
    s_i = \rank([A_1; \cdots ; A_{i}]) - \rank([A_1; \cdots ; A_{i-1}]).
    \]
    There is a matrix $U \in \FF^{k \times k}$ such that
\begin{equation}\label{eq:cspan2}
    \col([(A_1 U)|_{[s_1]} ; (A_2 U)|_{[s_2]} ; \cdots ; (A_m U)|_{[s_m]}]) = \col([A_1; \cdots ; A_m]).
\end{equation}
\end{lemma}

\begin{proof}
We begin by showing that there exists a matrix $U \in \FF^{k \times k}$ such that for all $i \in [m]$ we have 
\begin{equation}\label{eq:cspan1}
    \col([A_1 ; A_2 ; \cdots ; A_{i-1} ; (A_i U)|_{[s_i]}]) = \col([A_1; \cdots ; A_i]).
\end{equation}
From Lemma~\ref{lem:colspan}, we get that for each $i\in [m]$, there exists a nonzero~$k^2$-variable polynomial $p_i(X)\in \FF[x_{1,1},x_{1,2},\dots,x_{k,k}]$ of degree~$s_i$ such that for any $U\in \FF^{k\times k}$ satisfying~$p_i(U)\ne 0$,
\[
    \col([A_1 ; A_2 ; \cdots ; A_{i-1} ; (A_i U)|_{[s_i]}]) = \col([A_1; \cdots ; A_i]).
    \]
Let $q(X)=  p_1(X)\cdots p_m(X)$.
This is a nonzero polynomial of degree $s_1 + \cdots +s_m = \rank([A_1; \cdots; A_m])<|\FF|$.
It follows from the Schwartz--Zippel lemma \cite{DEMILLO1978193, schwartz1980fast, zippel1979probabilistic} that there is a~$U\in \FF^{k\times k}$ such that~$q(U) \ne 0$.
This is to say that for all~$i\in[m]$ simultaneously,~$p_i(U)\ne 0$, which implies the existence of the desired matrix~$U$.

Let~$U\in \FF^{k\times k}$ be a matrix so that~\eqref{eq:cspan1} holds for each $i\in [m]$.
Next, we show by induction on~$i\in[m]$ that~$U$ also satisfies~\eqref{eq:cspan2}.
The base case is $\col((A_1 U)|_{[s_1]}) = \col(A_1)$,
which follows directly.
For the induction step, we assume
\begin{equation}\label{eq:ih}
    \col([A_1 ; A_2 ; \cdots ; A_i]) = \col([(A_1U)|_{[s_1]} ; \cdots ; (A_iU)|_{[s_i]}]).
\end{equation}
We now use that if $\col(A) = \col(B)$, then $\col([A ; C]) = \col([B;C])$ holds for any matrices $A,B,C$ with the same number of rows. Applying this to \eqref{eq:ih} gives
\begin{equation}\label{eq:x}
    \col([A_1 ; A_2 ; \cdots ; A_i; (A_{i+1} U)|_{[s_{i+1}]}]) = \col([(A_1U)|_{[s_1]} ; \cdots ; (A_iU)|_{[s_i]}; (A_{i+1} U)|_{[s_{i+1}]}]).
\end{equation}
From \eqref{eq:cspan1} we know that 
\begin{equation}\label{eq:sz}
    \col([A_1; \cdots ; A_{i+1}]) = \col([A_1 ; A_2 ; \cdots ; A_i ; (A_{i+1} U)|_{[s_{i+1}]}]).
\end{equation}
Then \eqref{eq:x} and \eqref{eq:sz} together give
    \[
    \col([A_1 ; A_2 ; \cdots ; A_i; A_{i+1}]) = \col([(A_1U)|_{[s_1]} ; \cdots ; (A_iU)|_{[s_i]}; (A_{i+1}U)|_{[s_{i+1}]}]).
    \]
This proves the claim by induction.
\end{proof}

\begin{proof}[Proof of \autoref{th:uncertainty-principle-n}]
We will prove $n_1 \leq \subrank_2(T) \subrank_3(T)$.
Let $A_1, \ldots, A_{n_3} \in \FF^{n_1 \times n_2}$ be the $3$-slices of $T$.
Let~$U$ be as in \autoref{lem:max-rank-induction}.
Let 
\[
s_i = \rank([A_1; \cdots ; A_{i}]) - \rank([A_1; \cdots ; A_{i-1}]).
\]
Conciseness of~$T$ gives that $s_1 + \cdots + s_{n_3} = n_1$.
It follows from this and \autoref{lem:max-rank-induction} that the columns of the submatrices $(A_iU)|_{[s_i]}$ are linearly independent.
This implies that for each $i\in[n_3]$, 
there is a 3-slice of rank at least $s_i$,
and in particular $\subrank_3(T)\geq\max_i s_i$. %
Let $S$ be the tensor with 3-slices $A_1U, A_2U, \ldots, A_{n_3}U$. Then $\subrank_2(S) \leq \subrank_2(T)$ since the 2-slices of $S$ are linear combinations of the 2-slices of $T$. 
The matrix $M = [(A_1U)|_{[1]};\cdots ; (A_{n_3}U)|_{[1]}]$ consisting of every first column of the 3-slices of $S$ equals the first 2-slice of~$S$. Since the linearly independent columns in the 3-slices of~$S$ are flushed to the left, we find that $M$ has at least $|\{i\in [n_3]: s_i\ne 0\}|$ linearly independent columns, and thus $\rank(M) \geq |\{i\in [n_3]: s_i\ne 0\}|$,
which gives a lower bound on~$\subrank_2(S)$ and thus on $\subrank_2(T)$.
Since
\[
n_1  = s_1 + \cdots + s_{n_3}\leq |\{i\in [n_3]: s_i\ne 0\}|\, \max_i s_i \leq \subrank_2(T) \subrank_3(T),
\]
the result follows. 
\end{proof}

\newcommand{\yesTextColour}{green!60!black}
\newcommand{\noTextColour}{red!60!black}

\newcommand{\yesCellColour}{green!20}
\newcommand{\noCellColour}{red!20}

\subsection{Examples of optimality of max-rank bounds}
The following examples show that the bound in \autoref{th:uncertainty-principle-n} is tight. In particular, the right hand side cannot be super-linear in $n_k$, and in fact the leading coefficient $1$ is optimal. %

\begin{example}[{Null-algebra \cite[page~168]{strassen1991degeneration}}] \label{ex:strassenNullAlgebra}
For every $n \in \NN$,
there is a concise tensor $T \in \FF^n \otimes \FF^n \otimes \FF^n$ such that $\subrank_1(T) = \subrank_3(T) = n$ and $\subrank_2(T) = 2$. In particular, $\subrank_2(T) \subrank_3(T) = 2n$.
    To define this tensor, let
    \[T :=e_1\otimes e_1\otimes e_1 + \sum_{i=2}^n e_1\otimes e_i\otimes e_i + \sum_{i=2}^ne_i\otimes e_i\otimes e_1.\]
    We see directly that $\subrank_1(T) = \subrank_3(T) = n$. 
    Because the non-zero elements of the $2$-slices are confined to the first column and the first row, we have $\subrank_2(T)=2$. We also mention that for all $n\geq5$ the tensor $T$ has asymptotic subrank $\asympsubrank(T)=2\sqrt{n-1}$ \cite[page~169]{strassen1991degeneration}. %
\end{example}

\begin{example}%
    \label{ex:generalisedNullAlg}
    This example extends \autoref{ex:strassenNullAlgebra} (by taking $c=1$).
    Let $c$ be a constant and let $n$  be any multiple of $c$.
    There is a concise tensor $T \in \FF^{n} \otimes \FF^n \otimes \FF^n$ such that
    $\subrank_1(T)=n$, $\subrank_2(T)\leq c+1$ and $\subrank_3(T)\leq n/c+1$.
    In particular,
    \[
    \subrank_2(T)\subrank_3(T)\leq  \tfrac{(c+1)}{c} n+c+1.
    \]
    To define this tensor, let
    \begin{align*}
    T:=\sum_{i=1}^n e_1\otimes e_i\otimes e_i+&\sum_{k=1}^{c}\sum_{j=1}^{n/c}e_{(k-1)\frac{n}{c}+j}\otimes e_j\otimes e_k\\
    =\sum_{i=1}^n e_1\otimes e_i\otimes e_i+&\sum_{i=0}^{n-1} e_{i+1}\otimes e_{(i \bmod \frac{n}{c}) + 1}\otimes e_{\lfloor i / (\frac{n}{c})\rfloor+1}.
    \end{align*}
    We will prove that $\subrank_2(T)\leq c+1$ and $\subrank_3(T)\leq n/c+1$.
    Looking at direction $2$ (interpreting direction $1$ as enumerating the rows and $3$ as enumerating the columns), we can see all slices are restricted to the first row (first sum) or the first $c$ columns (the other sum) 
    for non-zero entries, so each matrix in their span, is of rank $\subrank_2(T)\leq c+1$. 
    Similarly, all slices in direction $3$ (interpreting direction $2$ as enumerating the rows) are restricted to the first $n/c$ rows or the first column, meaning $\subrank_3(T)\leq n/c+1$.

\end{example}

\begin{example}\label{ex:JopsEx}
    Let $n$ be a square.
    There is a concise tensor $T \in \FF^{n} \otimes \FF^n \otimes \FF^n$ such that for all~$\ell\in[3]$ we have
    $\sqrt{n} \leq \subrank_\ell(T)\leq 2\sqrt{n}$. In particular, for all ${\ell_1} \neq {\ell_2}$ we have $n\leq\subrank_{\ell_1}(T)\subrank_{\ell_2}(T)\leq4n$. 

    To define this tensor, let $f:[n]\to[\sqrt{n}]$ and $g:[n]\to[\sqrt{n}]$ be any two functions such that 
    the function $(f,g):[n]\to[\sqrt{n}]\times[\sqrt{n}] : i \mapsto (f(i), g(i))$ is injective and such that for every $i\in[\sqrt{n}]$ we have $f(i)=g(i)=i$. Let
    \[T:=\sum_{i=1}^{\sqrt{n}}e_i\otimes e_i\otimes e_i+\sum_{i=\sqrt{n}+1}^n e_i\otimes e_{f(i)}\otimes e_{g(i)} + e_{f(i)}\otimes e_i\otimes e_{g(i)} + e_{f(i)}\otimes e_{g(i)}\otimes e_i\]
    Then $T$ has a diagonal subtensor of size $\sqrt{n}$, so $\sqrt{n} \leq \subrank(T) \leq \subrank_\ell(T)$.
    On the other hand, $\subrank_\ell(T) \leq 2\sqrt{n}$, which can be seen using \autoref{th:Qi-SRi} 
    as the support of all its slices is confined to the first $\sqrt{n}$ rows and the first $\sqrt{n}$ columns, so $\slicerank_\ell(T)\leq 2\sqrt{n}$ (see \autoref{subsec:mincov} for the definition of $\slicerank_\ell$). From injectivity of $(f,g)$ it follows that $T$ is concise. %
\end{example}

\subsection{Lower bound on the asymptotic subrank via max-ranks and matrix multiplication tensors}\label{subsec:lb-mamu}

We will now prove a lower bound on the asymptotic subrank of concise tensors in terms of the dimensions of the tensor. We will use the max-rank inequality (\autoref{th:uncertainty-principle-n}) proven before and basic properties of the matrix multiplication tensors that we will discuss now. The main result here is the cube-root lower bound $\asympsubrank(T) \geq (\min_\ell n_\ell)^{1/3}$ for concise tensors $T \in \FF^{n_1} \otimes \FF^{n_2} \otimes \FF^{n_3}$ (\autoref{th:cube-root-lb}). With a simpler argument we can already obtain the bound $\asympsubrank(T) \geq (\min_\ell n_\ell)^{1/4}$ (\autoref{th:fourth-root}), which we will prove first. This fourth-root lower bound is already sufficient for our use in the proof of the discreteness theorem. In \autoref{sec:lower-bound-asympsubrank} we will prove an even better square-root lower bound under a symmetry assumption.

For $a,b,c \in \NN$ we denote by $\langle a,b,c\rangle \in \FF^{ab} \otimes \FF^{bc} \otimes \FF^{ca}$ the %
matrix multiplication tensor $\langle a,b,c \rangle = \sum_{i=1}^a \sum_{j=1}^b \sum_{k=1}^c e_{i,j} \otimes e_{j,k} \otimes e_{k,i}$. 
These tensors have the multiplicative property $\langle a,b,c\rangle \otimes \langle x,y,z\rangle = \langle ax, by, cz\rangle$. They are naturally related to the parameters~$\subrank_i$ as follows:

\begin{lemma}\label{lem:Qi-mamu}
Let $T$ be a tensor and $r \in \NN$. The following holds:
\begin{itemize}
\item $\subrank_1(T) \geq r$ if and only if $T \geq \langle 1, 1,r \rangle$.
\item $\subrank_2(T) \geq r$ if and only if $T \geq \langle r, 1,1 \rangle$.
\item $\subrank_3(T) \geq r$ if and only if $T \geq \langle 1, r, 1 \rangle$.  
\end{itemize}
\end{lemma}

\begin{lemma}\label{lem:Qi-square}
    If there are distinct ${\ell_1},{\ell_2}$ such that $\subrank_{\ell_1}(T), \subrank_{\ell_2}(T) \geq r$, then $\subrank(T^{\boxtimes 2}) \geq r$.
\end{lemma}

\begin{proof}
Suppose $\subrank_1(T), \subrank_2(T) \geq r$ (the proof for the other cases goes the same). Then we have that
$
T \geq \langle1, 1, r\rangle =  \sum_{i=1}^r e_1 \otimes e_i \otimes e_i$ and 
$T \geq \langle r, 1, 1\rangle = \sum_{i=1}^r e_i \otimes e_1 \otimes e_i$.
We multiply these inequalities to get
$T^{\boxtimes 2} \geq \langle r,1,r\rangle= \sum_{i,j=1}^r e_i \otimes e_j \otimes (e_i \otimes e_j).
$
We apply the projection $(e_i \otimes e_j) \mapsto \delta_{i=j} e_i$ to the third tensor leg to get
$
T^{\boxtimes 2} \geq \sum_{i=1}^r e_i \otimes e_i \otimes e_i
$
which proves the claim.
\end{proof}

We are now ready to prove the fourth-root lower bound on the asymptotic subrank. After that we will prove the stronger cube-root lower bound (\autoref{th:cube-root-lb}).

\begin{theorem}\label{th:fourth-root}
    Let $T \in \FF^{n_1} \otimes \FF^{n_2} \otimes \FF^{n_3}$ be concise. Then $\asympsubrank(T) \geq (\min_\ell n_\ell)^{1/4}$.
\end{theorem}
\begin{proof}
The proof is an application of \autoref{th:uncertainty-principle-n}.
That theorem only holds over large enough fields however. 
Since the asymptotic subrank does not change under field extension \cite[Theorem~3.10]{strassen1988asymptotic}, we may assume that our base field $\FF$ is infinite (by working over the algebraic closure of the original field) so that the field size condition of \autoref{th:uncertainty-principle-n} is satisfied. %
We use that there are distinct $\ell_1, \ell_2 \in [3]$ such that 
$\subrank_{\ell_1}(T), \subrank_{\ell_2}(T) \geq \sqrt{\min_\ell n_\ell}$ (\autoref{cor:twoLargeRankDirs}).
Then $\asympsubrank(T^{\boxtimes 2}) \geq \sqrt{\min_\ell n_\ell}$ (\autoref{lem:Qi-square}) which gives the claim. %
\end{proof}

It turns out that the above construction is not optimal. We will now give a slightly more elaborate construction that leads to the better cube-root bound:

\begin{theorem}\label{th:cube-root-lb}
    Let $T \in \FF^{n_1} \otimes \FF^{n_2} \otimes \FF^{n_3}$ be concise. Then $\asympsubrank(T) \geq (\min_{\ell} n_\ell)^{1/3}$.
\end{theorem}

The proof of \autoref{th:cube-root-lb} makes use of basic properties of matrix multiplication tensors which we discuss in the following two lemmas. 

\begin{lemma}\label{lem:cube-mamu}
    $T^{\boxtimes 3} \geq \langle \subrank_2(T),\subrank_3(T),\subrank_1(T) \rangle.$
\end{lemma}

\begin{proof}
    From \autoref{lem:Qi-mamu} we have 
    $T \geq \langle \subrank_2(T),1,1\rangle$, 
    $T \geq \langle 1,\subrank_3(T),1\rangle$, and  
    $T \geq \langle 1,1,\subrank_1(T)\rangle$.
We multiply these inequalities using the Kronecker product to get the claim.
\end{proof}

\begin{lemma}\label{lem:asympsub-min}
    $\asympsubrank(T^{\boxtimes 3}) \geq \min_{{\ell_1},{\ell_2}} \subrank_{\ell_1}(T) \subrank_{\ell_2}(T).$
\end{lemma}
\begin{proof}
    It is known that for every $a,b,c \in \NN$ the asymptotic subrank of the matrix multiplication tensor $\langle a,b,c\rangle$ is given by $\asympsubrank(\langle a,b,c\rangle) = \min \{ab, bc, ac\}$ \cite[Equation~4.6]{strassen1988asymptotic}. We combine this with \autoref{lem:cube-mamu} to get the claim.
\end{proof}

\begin{proof}[Proof of \autoref{th:cube-root-lb}]
As in the proof of \autoref{th:fourth-root} we may assume that the base field $\FF$ is infinite.
We use $\asympsubrank(T^{\boxtimes 3}) \geq \min_{{\ell_1},{\ell_2}} \subrank_{\ell_1}(T) \subrank_{\ell_2}(T)$ (\autoref{lem:asympsub-min}) and we use that for every distinct ${\ell_1}, {\ell_2}, {\ell_3} \in [3]$ we have $\subrank_{\ell_1}(T)\subrank_{\ell_2}(T) \geq n_{\ell_3}$ (\autoref{th:uncertainty-principle-n}) to obtain
$\asympsubrank(T^{\boxtimes 3}) \geq \min_{{\ell_1},{\ell_2}} \subrank_{\ell_1}(T) \subrank_{\ell_2}(T) \geq \min_{\ell_3} n_{\ell_3}$. %
\end{proof}

\section{Min-rank of slice spans and asymptotic subrank of narrow tensors} %
\label{sec:lower-bound-imbalanced}

In \autoref{sec:max-rank} %
we lower bounded the asymptotic subrank of a concise tensor by a growing function in the smallest local dimension (\autoref{th:cube-root-lb}). Namely, we showed that for every concise $n_1 \times n_2 \times n_3$ tensor the asymptotic subrank is at least $(\min_\ell n_\ell)^{1/3}$. 
We will use that result as an ingredient to prove the discreteness theorem in \autoref{sec:discreteness}. However, that bound will not be enough as we must also understand how asymptotic subrank behaves when the smallest dimension $\min_\ell n_\ell$ is fixed (we elaborate on this point in \autoref{sec:discreteness}). This is the topic of this section.

To briefly explain what is missing, we note that it is a priori possible that there is a sequence of concise $c \times n_1 \times n_2$ tensors with $c$ constant and at least one of $n_1, n_2$ growing\footnote{We note that in this context, since we consider only concise tensors, if $c$ is constant and at least one of $n_1, n_2$ is growing, then both $n_1$ and $n_2$ must be growing by \autoref{cl:concise-formats}.}
such that the asymptotic subrank (or asymptotic slice rank) has an accumulation point. Our result will rule this out in the appropriate sense. In particular, we prove (\autoref{theorem:concise-cn1n2-assubrank-c}) that for every $c \geq 2$ there is a number $N(c)$ such that for every $n_1, n_2$ with $\max \{n_1, n_2\} > N(c)$,
for every concise $c \times n_1 \times n_2$ tensor, the asymptotic subrank equals $c$. This will be a crucial ingredient in our proof of the discreteness theorem in \autoref{sec:discreteness}.

Whereas the proof of the asymptotic subrank lower bound in \autoref{sec:max-rank} relied for a large part on the max-rank of matrix subspaces, here we will need the related notion of min-rank of a matrix subspace.
In \autoref{subsec:minrank} we introduce the notion of min-rank, which we will use throughout. In \autoref{subsec:minrank-subrank} we prove that if the slices of a concise tensor have large min-rank, then the tensor has large subrank. In \autoref{subsec:maxrank-minrank} we see how from matrices with large max-rank we can construct matrices with large min-rank that are moreover diagonal. In \autoref{subsec:supermult} we then prove that min-rank is super-multiplicative when at least one of the matrix subspaces is diagonal, and finally we generalize the super-multiplicativity property to make it work with the output of \autoref{subsec:maxrank-minrank}.

In \autoref{subsec:main-unb} we combine the lemmas on min-rank to prove the main result (\autoref{theorem:concise-cn1n2-assubrank-c}). Finally, as an extra, in \autoref{subsec:dimtwo}, we determine precisely what happens in the smallest case $c = 2$, and in particular determine that $N(2) = 2$.

To avoid potential confusion we remark that throughout we will call any matrix $A\in \FF^{n_1} \otimes \FF^{n_2}$ (that is not necessarily square) \emph{diagonal} if its support, which is defined as $\supp(A) := \{(i,j) : A_{i,j} \neq 0\}$, is a subset of $\{(i,i): 1\leq i\leq \min\{n_1, n_2\}\}$.

\subsection{Min-rank of matrix subspaces and its properties}\label{subsec:minrank} %
\newcommand{\StrExC}{T_{\textnormal{NA},c,n}}
The min-rank of a matrix subspace is simply the smallest rank of any non-zero matrix in the space.

\begin{definition}
    We define the \emph{min-rank} of a matrix subspace $\mathcal{A} \subseteq \FF^{n_1} \otimes \FF^{n_2}$ as
    \[
    \minrnk(\mathcal{A}):=\min\{\rank(A):A\in\mathcal{A}\setminus\{0\}\}.
    \]
    For any collection of matrices $A_1,
    \ldots,A_c \in \FF^{n_1} \otimes \FF^{n_2}$ we define \[
    \minrnk(A_1,\dots,A_c):=\minrnk(\linspan\{A_1,\dots,A_c\}).
    \]
\end{definition}

We will use the following straightforward equivalence. 

\begin{lemma}\label{lem:minrank-nontrivial-lin-comb}
    Let $A_1,\dots,A_c \in \FF^{n_1} \otimes \FF^{n_2}$ be matrices. %
    For every $r\geq1$, the following two statements are equivalent:
    \begin{enumerate}[\upshape(a)]
        \item The matrices $A_1,\dots,A_c$ are linearly independent and $\minrnk(A_1,\dots,A_c)\geq r$.
        \item For every $u \in \FF^c \setminus \{0\}$ we have $\rank(u_1 A_1 + \cdots + u_c A_c) \geq r$.
    \end{enumerate} 
\end{lemma}
\subsubsection{Large min-rank implies large subrank}\label{subsec:minrank-subrank}

In this section we prove that if the slices of a tensor have large min-rank, then the tensor has large subrank, in the following sense:


\begin{restatable}{lemma}{lemD}
\label{cla-high-rank-high-subrank}
Let $T\in\FF^{n_1}\otimes\FF^{n_2}\otimes\FF^{n_3}$ be a tensor and $A_1,\ldots,A_c\in\FF^{n_1}\otimes\FF^{n_2}$ be $c$ 3-slices of it. %
If $A_1,\ldots,A_c$ are linearly independent and  $\minrnk(A_1,\dots,A_c)\geq 2c(c-1)$,
then $\subrank(T)\geq c$.    
\end{restatable}
%
%
%
%

%
%
%

%
%
%
%


%

%

%

%
%

%
%
%

\begin{proof}
%
For $i \in [c]$ and $\alpha \in \{0,1\}$, we say that an $m_1\times m_2$ matrix $A$ has property $P(i,\alpha)$ if $A_{i,i} = \alpha$ and $A_{j, i} = 0$ for all $j \in [m_1]\setminus \{i\}$ and $A_{i,j} = 0$ for all $j \in [m_2] \setminus \{i\}$.

We may assume $n_3 = c$.
We will be changing $T$ with a series of restrictions after which the slices $A_1, \ldots, A_c$ will satisfy: for every $j \in [c]$ the matrix $A_j$ has property $P(j, 1)$ and for every $k \in [c] \setminus \{j\}$ the matrix $A_k$ has property $P(j, 0)$. Then $A_1, \ldots, A_c$ restricted to the first $c$ rows and columns form the unit tensor $\langle c\rangle$.

We will carry out the above inductively, obtaining by induction on $0 \leq i \leq c$, that 
\begin{itemize}
\item $A_1, \ldots, A_c$ are linearly independent and $\minrank(A_1, \ldots, A_c) \geq (c-i)\cdot 2(c-1)$
\item for all $1 \leq j \leq i$, 
\begin{itemize}
    \item $A_j$ has property $P(j, 1)$
    \item for all $k \in [c] \setminus \{j\}$, $A_k$ has property $P(j, 0)$.
\end{itemize}  
\end{itemize}
Throughout the induction, the restrictions we use may change the matrices $A_1, \ldots, A_c$, but we will keep denoting them by these names. In particular, the matrices $A_1, \ldots, A_c$ may become smaller. Namely, in every step of the induction we will remove at most $c-1$ rows and at most $c-1$ columns from all matrices (simultaneously). Note that the case $i=0$ is the assumption of the lemma.

Let $1 \leq i\leq c$ and assume the induction claim holds for $i-1$. 
In the following, all operations on rows and columns we apply to all matrices at the same time.
By linear independence, $A_i$ is nonzero. For $j < i$, row $j$ and column $j$ of $A_i$ are zero (by induction hypothesis), so there are $j,k \geq i$ with $(A_i)_{j,k}$ nonzero.
Swap rows $i$ and $j$ and swap columns~$i$ and $k$ and scale so that $(A_i)_{i,i} = 1$. By adding scalar multiples of row $i$ to rows $j > i$ and adding scalar multiples of column $i$ to columns $k > i$ we make row $i$ and column~$i$ in $A_i$ zero. Then $A_i$ has property $P(i,1)$. 

Now we go over $j \in [c] \setminus \{i\}$ one by one. By adding a scalar multiple of $A_i$ to $A_j$ we get $(A_j)_{i,i} = 0$. Suppose the $i$th row of $A_j$ is nonzero. Then there is a $k>i$ such that $(A_j)_{i,k}$ is nonzero (by induction hypothesis). Add scalar multiples of the $k$th column to the other columns $k'>i$ to make $(A_j)_{i,k'}$ zero. Then remove the $k$th column. Then the $i$th row of $A_j$ is zero. Repeat the same procedure for the $i$th column of $A_j$ to make it zero, possibly removing one row. As a result, $A_j$ has property $P(i,0)$.

For every $j \in [c] \setminus \{i\}$, we removed at most one row and one column, so in total we removed at most $2(c-1)$ rows or columns. Removing a row or column decreases matrix rank by at most one, so (\autoref{lem:minrank-nontrivial-lin-comb})
for any $u\in\FF^{c}\setminus\{0\}$ we have $\rank(u_1A_1+\cdots+u_cA_c)\geq
(c-(i-1))\cdot 2(c-1) - 2(c-1) = (c-i)\cdot 2(c-1)$. 
If $i < c$, this is nonzero so $A_1, \ldots, A_c$ are linearly independent and $\minrank(A_1,\ldots,A_c) \geq (c-i)\cdot 2(c-1)$ (\autoref{lem:minrank-nontrivial-lin-comb}). (This is clear for $i=c$ since then the matrices form the tensor $\langle c\rangle$.) 
This proves the induction claim for $i$.
\end{proof}

\subsubsection{From large max-rank to large min-rank and diagonal}\label{subsec:maxrank-minrank}

In this section we discuss how from any collection of matrices that contains \emph{at least one} matrix of large rank (large max-rank) we can produce a collection of slightly smaller matrices such that \emph{all} nonzero matrices in their span have large rank (large min-rank). Moreover, the resulting matrices are all diagonal, which will be crucial in our application.

For any matrix $A\in \FF^{n_1} \otimes \FF^{n_2}$ %
and subsets $I\subseteq [n_1]$, $J\subseteq [n_2]$ we denote by $A_{I\times J}$ the submatrix of $A$ with rows indexed by $I$ and columns indexed by~$J$.  For any subset $X \subseteq \FF^{n_1} \otimes \FF^{n_2}$ we denote by $X_{I \times J}$ the set $\{A_{I\times J} : A \in X\}$.

\begin{restatable}{lemma}{lemBx}
\label{lem:minrk_diag}
\label{claim:matrices-into-diagonal-with-large-rank}
For every $c\in\NN$, there is an $\eps(c)>0$ such that the following holds.
Let $\mathcal A\subseteq \FF^{n_1\times n_2}$ be a matrix subspace of dimension~$c$.
Then, there exist $U\in \GL(\FF^{n_1})$, $V\in \GL(\FF^{n_2})$, a subset $J\subseteq [\min_\ell n_\ell]$, a basis $B_1,\dots,B_c$ for the space $U\mathcal AV$ and $b\in[c]$ such that:
\begin{itemize}
    \item $(B_1)_{J\times J},\dots,(B_b)_{J\times J}$  are diagonal and linearly independent;
    \item $\minrank\big((B_1)_{J\times J}, \ldots, (B_b)_{J\times J}\big) \geq \eps(c) \maxrank(\mathcal{A})$; 
    \item $(B_{b+1})_{J\times J} = \cdots = (B_c)_{J\times J} = 0$.
\end{itemize}
\end{restatable}

Before discussing the proof in detail, as a motivation and ingredient for later, we observe using the pigeonhole principle that any concise tensor has at least one slice of large rank. 

\begin{restatable}{lemma}{lemA}
\label{cla:high-rk}
Let $T\in\FF^{n_1}\otimes\FF^{n_2}\otimes\FF^{c}$ be a concise tensor. Then $T$ has at least one 3-slice 
of rank at least $\max\{n_1,n_2\}/c$.    
\end{restatable}

\begin{proof}
Suppose $n_1\geq n_2$. %
Concatenate the 3-slices $A_1, \ldots, A_c$ of $T$ to an $n_1\times cn_2$ matrix. Conciseness of $T$ implies that the rank of this matrix is $n_1$ so there are $n_1$ linearly independent columns. By the pigeonhole principle, there is an $i \in [c]$ such that $A_i$ contains at least $n_1/c$ linearly independent columns.   
\end{proof}

\paragraph{Diagonalization.}
The proof of \autoref{lem:minrk_diag} has three parts.
In this first part, we show in \autoref{lem:diagonalize} that any collection of matrices can be made simultaneously diagonal on a large principal submatrix in such a way that the property of having large max-rank is preserved. 
For $a,b \in [n]$ let $E_{a,b} \in \FF^n \otimes \FF^n$ be the matrix with coefficient one at index $(a,b)$ and coefficient zero elsewhere.

\begin{lemma}\label{lem:diagonalize-ind2}
    Let $A \in \FF^{n \times n}$. Then there are $U,V \in \GL(\FF^n)$ and $i,j \in \{2,\ldots, n\}$ such that for $I = [n] \setminus\{i,j\}$,
    \begin{enumerate}[\upshape (i)]
    \item $(UAV)_{I\times I}  = \alpha_{1,1} E_{1,1} + \sum_{a,b \geq 2} \alpha_{a,b} E_{a,b}$ for some $\alpha_{a,b} \in \FF$;
    \item for every diagonal matrix $B\in \FF^{n\times n}$, we have that $(U B V)_{I\times I} = B_{I\times I}$.
    \end{enumerate}
\end{lemma}
\begin{proof}
    Let $X = \{\ell \in \{2, \ldots, n\} : A_{1, \ell} \neq 0\}$. If $X = \emptyset$, then we let $V = \id$ and let~$j$ be any element of $[n]$. Otherwise, we let $j$ be any element of $X$, and we let $V$ be the matrix that, when multiplied with $A$ from the right, adds multiples of the $j$th column of $A$ to the $\ell$th column of $A$ such that the $(1, \ell)$ entry becomes zero, for every $\ell \in X \setminus \{j\}$%
    . In other words, we let
    \[
    V = \id - \sum_{\mathclap{\ell \in X \setminus \{j\}}} A_{1,j}^{-1} A_{1,\ell} E_{j, \ell}.
    \]
    Let $Y = \{\ell \in \{2, \ldots, n\} : A_{\ell, 1} \neq 0\}$. If $Y= \emptyset$ or $Y = \{j\}$, then we let $U = \id$ and we let $i$ be any element of $[n] \setminus \{j\}$. Otherwise, let $i$ be any element of $Y \setminus \{j\}$, and we let $U$ be the matrix that, when multiplied with $A$ from the left, add multiples of the $i$th row of $A$ to the $\ell$th row of $A$ such that the $(\ell, 1)$ entry becomes zero, for every $\ell \in Y \setminus \{i,j\}$.
    In other words,
    \[
    U = \id - \sum_{\mathclap{\ell \in Y \setminus \{i,j\}}} A_{i,1}^{-1} A_{\ell,1} E_{\ell, i}.
    \]
    Then claim (i) is satisfied. 
    To see that (ii) is satisfied we note that $U = \id + U_0$ and $V = \id + V_0$, where~$U_0$ supported by the $i$th column and of~$V_0$ is supported by the $j$th row.
    Then, for any diagonal matrix $B \in \FF^{n\times n}$,
    \begin{align*}
        UBV &= B + U_0B + BV_0 + U_0BV_0.
    \end{align*}
    Since $i\ne j$, it follows that $U_0BV_0 = 0$.
    Moreover, $U_0B$ is supported by the $i$th column and $BV_0$ is supported by the $j$th row, so $(U_0B)_{I \times I}=0$ and $(BV_0)_{I \times I}=0$, and so $(UBV)_{I \times I} = B_{I\times I}$.
\end{proof}

\begin{lemma}\label{lem:diagonalize-ind}
    Let $A \in \FF^{n \times n}$. Then there are $U,V \in \GL(\FF^n)$ and $I \subseteq [n]$ such that $|I| \geq n/3$ and 
    \begin{enumerate}[\upshape (i)]
    \item $(UA V)_{I \times I}$ is diagonal;
    \item for every diagonal matrix $B \in \FF^{n \times n}$,  we have that $(U B V)_{I \times I} = B_{I \times I}$.
\end{enumerate}
\end{lemma}
\begin{proof}
    The proof is by repeatedly applying \autoref{lem:diagonalize-ind2}.
    Applying \autoref{lem:diagonalize-ind2} to $A$ gives $U,V,I$ such that $(UAV)_{I \times I}$ is of the form 
    \[
    \alpha_{1,1} E_{1,1} + \sum_{a,b\geq 2} \alpha_{a,b} E_{a,b}
    \]
    and such that for any diagonal $B$ we have $(UBV)_{I \times I} = B_{I \times I}$.
    Suppose we are given $A' \in \FF^{m \times m}$ of the form 
    \[
    \sum_{i=1}^k \alpha_{i,i} E_{i,i} + \sum_{a,b\geq k+1} \alpha_{a,b} E_{a,b}.
    \]
    Applying \autoref{lem:diagonalize-ind2} to $A'' = (A')_{\{k+1, \ldots m\} \times \{k+1, \ldots m\}}$ gives $U', V', I'$ such that $(U' A'' V')_{I' \times I'}$ is of the form $\alpha'_{1,1} E_{1,1} + \sum_{a,b\geq 2} \alpha'_{a,b} E_{a,b}$. Then 
    \[
    ((\id_k\oplus U') A' (\id_k\oplus V'))_{I' \times I'}
    \]
    is of the form 
    \[
    \sum_{i=1}^{k+1} \alpha_{i,i} E_{i,i} + \sum_{a,b\geq k+2} \alpha_{a,b} E_{a,b}
    \]
    and for any diagonal $B$ we have $((\id_k\oplus U')UBV(\id_k\oplus V'))_{I' \times I'} = B_{I' \times I'}$. In every step of the process $k$ goes up by one, and $|I|$ goes down by two, so that in the end we obtain the claim with $I \geq n/3$.
\end{proof}

\begin{lemma}\label{lem:diagonalize}
Let $A_1 = \id_n, A_2, \ldots, A_c \in \FF^{n \times n}$. Then there are $U,V \in \GL(\FF^n)$ and $I \subseteq [n]$ such that $|I| \geq 3^{-(c-1)} n$ and
\begin{enumerate}[\upshape (i)]
    \item $(U \id_n V)_{I\times I} = (\id_n)_{I\times I}$
    \item $(UA_i V)_{I \times I}$ is diagonal for all $i \in [c]$.
\end{enumerate}
\end{lemma}
\begin{proof}
    The proof is by induction over $c$ using \autoref{lem:diagonalize-ind}. 
    The base case $c = 1$ is clear.
    For the induction step, we are given $A_1, \ldots, A_m$, $U,V$ and $I \subseteq[n]$ with $|I| \geq 3^{-(m-1)}n$ such that for $B_i = (UA_i V)_{I \times I}$ we have $B_1 = \id_I$ and every $B_i$ is diagonal.
    Let $B_{m+1} = (U A_{m+1} V)_{I \times I}$. We apply \autoref{lem:diagonalize-ind} to $A = B_{m+1}$. This gives $U', V' \in \GL(\FF^I)$ and $J \subseteq I$ such that $|J| \geq \tfrac13 |I|$, $(U' B_{m+1} V')_{J \times J}$ is diagonal, and $(U' B_i V')_{J \times J} = (B_i)_{J \times J}$ for every $i \in [m]$. This proves the claim.
\end{proof}

\paragraph{From large max-rank to large min-rank.}
In this second part of the proof of \autoref{lem:minrk_diag}, we show in \autoref{lem:large-minrank} that any low-dimensional matrix subspace~$\mathcal A$ of diagonal matrices with large max-rank contains a principal-matrix subspace with large min-rank.
Moreover, in \autoref{lem:basis-extension-matrices}, we show that any matrix subspace has a basis that, when restricted to a particular principal submatrix, consists of a set of linearly independent matrices and all-zero matrices.

For any vector $w \in \FF^n$ let $\supp(w) = \{i \in [n]: w_i \neq 0\}$ and let $|w| = |\supp(w)|$. For any subspace $W \subseteq \FF^n$ let
\begin{align*}
\maxsupp(W) &= \max \{ |w| : w \in W\},\\
\minsupp(W) &= \min \{ |w| : w \in W \setminus \{0\}\}.
\end{align*}

\begin{lemma}\label{lem:minsupp}
Let $V \subseteq \FF^n$ be a subspace such that $\dim(V) \leq c$.  
There is an $I \subseteq [n]$ such that $V_I$ is not zero and
\[
\minsupp(V_I) \geq \tfrac1c \maxsupp(V).
\]
\end{lemma}
\begin{proof}
    We will use the general fact that for arbitrary vectors $w_1, \ldots, w_\ell \in W$, if for every $j\in \{0, \ldots, \ell-1\}$ we have $\cup_{i \leq j} \supp(w_i) \neq \cup_{i \leq j + 1} \supp(w_i)$, then $w_1, \ldots, w_\ell$ are linearly independent.

    Let $k = \maxsupp(V)$ and $d = \dim(V)$.
    Before constructing $I$, we first construct a set $X \subseteq V$ by the following process, starting with $X$ empty.
    If there is a nonzero element $v_1 \in V$ such that $|v_1| < \tfrac1c k$, then add $v_1$ to $X$.
    If $X = \{v_1, \ldots, v_\ell\}$ and there is an element $v_{\ell + 1} \in V$ such that 
    $0 < |\cup_{i \leq \ell+1} \supp(v_i) \setminus \cup_{i \leq \ell} \supp(v_i)| < \tfrac1c k$, then add $v_{\ell+1}$ to $X$.
    By the aforementioned general fact, $X$ is linearly independent and thus $X$ is finite, namely $|X| \leq d$.

    
    We let $J = \cup_i \supp(v_i)$ and $I = [n] \setminus J$.
    
    First we claim that $V_I$ is not zero. We have $|J| < d\tfrac1c k \leq k$.
    Then for any element $v \in V$ with $|v| = k$ we have $|v_I| \geq k - |J| > 1$ and thus $v_I \neq 0$.
    
    Second, we claim $\minsupp(V_I) \geq \tfrac1c k$.  Suppose there is a nonzero element $w \in V_I$ such that $|w| < \tfrac1c k$, then this contradicts the way we constructed $X$.
\end{proof}

\begin{lemma}\label{lem:basis-extension}
    Let $V \subseteq \FF^n$ be a subspace of dimension $c$, and let $I \subseteq [n]$. Then there is a basis $v_1, \ldots, v_c$ for $V$ such that 
    \begin{itemize}
        \item $(v_1)_I, \ldots, (v_b)_I$ are linearly independent
        \item $(v_{b+1})_I = \cdots = (v_c)_I = 0$
    \end{itemize}
    for $b = \dim(V_I)$.
\end{lemma}
\begin{proof}
    Choose any basis $w_1, \ldots, w_c$ for $V$. Let $M$ be the matrix with these basis vectors as columns. The submatrix $M_{I \times[n]}$ has rank $b$ so there is a subset $J \subseteq [n]$ of size $b$ such that $M_{I \times J}$ has rank $b$.
    Let $\{v_1, \ldots, v_b\} = \{w_j : j \in J\}$ and $\{v_{b+1}, \ldots, v_c\} = \{w_j : j \not\in J\}$.
    Then $(v_1)_I, \ldots, (v_b)_I$ form a basis of $V_I$. By subtracting multiples of $v_1, \ldots, v_b$ from the vectors $v_{b+1}, \ldots, v_c$ we can ensure that $(v_{b+1})_I = \cdots = (v_c)_I = 0$.
\end{proof}

\begin{lemma}\label{lem:large-minrank}
    Let $\mathcal{A} \subseteq \FF^{n \times n}$ be a subspace of diagonal matrices such that $\dim(\mathcal{A}) \leq c$. There is an $I \subseteq [n]$ such that 
    \[
    \minrank(\mathcal{A}_{I\times I}) \geq \tfrac1c \maxrank(\mathcal{A}).
    \]
\end{lemma}
\begin{proof}
    This follows from \autoref{lem:minsupp} by identifying every diagonal matrix $A \in \mathcal{A}$ with the vector of its diagonal coefficients $v = (A_{1,1}, \ldots, A_{n,n})$ and $\rank(A) = |\supp(v)|$.
\end{proof}

\begin{lemma}\label{lem:basis-extension-matrices}
    Let $\mathcal{A} \subseteq \FF^{n_1 \times n_2}$ be matrix subspace of dimension $c$, and $J \subseteq [\min_\ell n_\ell]$. Then there is a basis $B_1, \ldots, B_c$ for $\mathcal{A}$ such that
    \begin{itemize}
        \item $(B_1)_{J \times J}, \ldots, (B_b)_{J \times J}$ are linearly independent
        \item $(B_{b+1})_{J \times J}= \cdots = (B_{c})_{J \times J} = 0$
    \end{itemize}
    for $b = \dim(\mathcal{A}_{J\times J})$.
\end{lemma}
\begin{proof}
This follows from applying \autoref{lem:basis-extension} to $V = \mathcal{A}$ and $I = J \times J$.
\end{proof}

\paragraph{Combining the above.} \!We now combine \autoref{lem:diagonalize}, \autoref{lem:large-minrank} and \autoref{lem:basis-extension-matrices} to prove \autoref{lem:minrk_diag}.

\begin{proof}[Proof of \autoref{lem:minrk_diag}]
Let $k = \maxrank(\mathcal A)$.
    By Lemma~\ref{cla:high-rk}, we have that $k \geq n_2/c$.
    After a basis transformation, we may assume that $\mathcal A$ has a basis $A_1 = \id_k,A_2,\dots,A_c$.
    Applying \autoref{lem:diagonalize} to the submatrices $(A_i)_{[k]\times [k]}$, we obtain matrices $U,V\in \GL(\FF^k)$ and a set $I\subseteq[k]$ of size at least $3^{-(c-1)}k$ such that
    \begin{enumerate}[(i)]
        \item $A_1' = \big((U\oplus \id_{n_1 - k})A_1(V\oplus \id_{n_2 - k})\big)_{I\times I} = (\id_k)_{I\times I}$;
    \item $A_i' = \big((U\oplus \id_{n_1 - k})A_i(V\oplus \id_{n_2 - k})\big)_{I\times I}$ is diagonal for all $i\in [c]$.
    \end{enumerate}
    By \autoref{lem:large-minrank} there is a set~$J\subseteq I$ such that
    \beqn
    \minrnk\big((A_1')_{J\times J},\dots,(A_c')_{J\times J}\big) \geq \tfrac{1}{c} 3^{-(c-1)}k. 
    \eeqn
    The claim follows from applying \autoref{lem:basis-extension-matrices} with $B_i = A'_i$.
\end{proof}

\subsubsection{Super-multiplicativity of min-rank}\label{subsec:supermult}

For subspaces of matrices $\mathcal{A}$ and $\mathcal{B}$ their tensor product $\mathcal{A} \otimes \mathcal{B}$ is the subspace of matrices spanned by the products $A \otimes B$ for all $A \in \mathcal{A}$ and $B \in \mathcal{B}$.
In this section we will prove that the min-rank is super-multiplicative under the tensor product, under the assumption that at least one of the matrix subspaces is \emph{diagonal}. (Without this assumption the min-rank can be strictly submultiplicative, see \autoref{ex:minrank-not-supermultiplicative}.)
We then apply this to the min-rank of powers of any diagonal subspace of matrices.

%
%

\begin{lemma}\label{lem:minrank-diag-supermultiplicative}
    Let $\mathcal{A}$ and $\mathcal{B}$ be subspaces of matrices. Suppose that all matrices in $\mathcal{A}$ are diagonal.
    Then %
$\minrnk(\mathcal{A}\otimes\mathcal{B})\geq\minrnk(\mathcal{A})\cdot\minrnk(\mathcal{B})$.
\end{lemma}
\begin{proof}%
    Let $C\in\mathcal{A}\otimes\mathcal{B}$ be a nonzero element. Then $C$ is of the form $C =\sum_{i=1}^\ell A_i\otimes B_i$ for some $A_i \in \mathcal{A}$, $B_i \in \mathcal{B}$ and $\ell \geq 1$. We may assume $\{A_i\}_{i=1}^\ell$ and $\{B_i\}_{i=1}^\ell$ are both linearly independent sets of matrices.  %

    Since every $A_i$ is diagonal, the matrix $C$ is block-diagonal with blocks $D_j$ given by $D_j = \sum_{i=1}^\ell (A_i)_{j,j} B_i$. Thus $\rank(C) = \sum_j \rank(D_j)$. For any $j$ such that $(A_1)_{j,j} \neq 0$, the matrix $D_j$ is a nontrivial linear combination of the $B_i$ and thus $\rank(D_j) \geq \minrnk(\mathcal{B})$. The number of $j$ such that $(A_1)_{j,j} \neq0$ is equal to $\rank(A_1)$ since $A_1$ is diagonal. %
    Thus
    \[
    \rank(C) = \sum_j \rank(D_j) \geq \sum_{\mathclap{j: (A_1)_{j,j} \neq 0}} \rank(D_j) %
    \geq \rank(A_1)\minrnk(\mathcal{B}) \geq \minrnk(\mathcal{A})\minrnk(\mathcal{B}),
    \]
    which proves the claim. %
\end{proof}
%
\begin{example}\label{ex:minrank-not-supermultiplicative}
    The condition in \autoref{lem:minrank-diag-supermultiplicative} that at least one of the matrix subspaces consists of diagonal matrices cannot be removed. To see this we may take the subspace of matrices over the reals
    \[
\mathcal{A} = \Bigl\{\begin{pmatrix}
    a & b\\-b & a
\end{pmatrix} : a,b \in \RR\Bigr\}.
    \]
    Then $\minrnk(\mathcal{A}) = 2$. However, $\mathcal{A}^{\otimes 2}$ contains the element
    \[
    \begin{pmatrix}
        1 & 0\\0 & 1
    \end{pmatrix}^{\otimes 2}
    +
    \begin{pmatrix}
        0 & 1\\-1 & 0
    \end{pmatrix}^{\otimes 2}
     = \begin{pmatrix}
         1 & 0 & 0 & 1\\
         0 & 1 & -1 & 0\\
         0 & -1 & 1 & 0\\
         1 & 0 & 0 & 1
     \end{pmatrix}
    \]
    which has rank $2 < 2^2$, so $\minrnk(\mathcal{A}^{\otimes 2}) < \minrnk(\mathcal{A})^2$.
\end{example}

\begin{lemma}\label{claim:large-rank-lin-comb-tensorprods}
    Let $\mathcal{A}$ %
    be a subspace of matrices. Suppose that all matrices in $\mathcal{A}$ are diagonal. %
    Then for all $k\in\NN$ we have $\minrnk(\mathcal{A}^{\otimes k})\geq \minrnk(\mathcal{A})^k$.
\end{lemma}
\begin{proof}
    Repeatedly apply \autoref{lem:minrank-diag-supermultiplicative}.
\end{proof}
In the rest of this section we prove a slightly more general version of \autoref{claim:large-rank-lin-comb-tensorprods}.
The reason that we need to do this is that we want to take the output of \autoref{claim:matrices-into-diagonal-with-large-rank} (a collection of linearly independent matrices, such that restricted to specified rows and columns the matrices are diagonal and have large min-rank, or are zero) and use the super-multiplicativity property of \autoref{claim:large-rank-lin-comb-tensorprods} to construct a large collection of matrices (that are tensor products of the original matrices) with large min-rank. However, we cannot do this directly as some of the matrices after restriction are zero. The lemma we prove here deals with that.


%
%

%
\begin{restatable}{lemma}{lemC}
\label{claim:linear-combinations-of-mixed-tensorprods-high-rank}
    Let $X = \{B_1, \ldots, B_b, C_{b+1} \ldots, C_c\} \subseteq \FF^{n_1} \otimes \FF^{n_2}$ be a set of linearly independent matrices. Let $I\subseteq [n_1]$, $J \subseteq [n_2]$, $|I|=|J|$, such that 
    \begin{itemize}
        \item the matrices $(B_i)_{I\times J}$, for $i = 1, \ldots, b$, are diagonal and linearly independent,
        \item the matrices $(C_i)_{I\times J}$, for $i = b+1, \ldots, c$, are zero.
    \end{itemize}
    For any $m \geq \ell\geq 1$, 
    let $Y$ %
    be the set of all order-$m$ Kronecker products of elements of~$X$ such that at least $\ell$ of the factors are from $B_1, \ldots, B_b$. Then %
    \[
    \minrnk(Y) \geq \minrnk(X_{I\times J})^{\ell}.
    \]
\end{restatable}

\begin{proof}
    Let $A$ be any nontrivial linear combination of elements of $Y$.
    Then
    \[
    A=\sum_{t=1}^r a_t\, D_{i_{t,1}}\otimes\cdots\otimes D_{i_{t,m}}
    \]
    for some nonzero $a_t \in \FF$ and $D_{i_{t, j}} \in X$. Every $D_{i_{t, j}}$ is either in $\{B_1, \ldots, B_b\}$, in which case we will call it a $B$-matrix, or in $\{C_{b+1}, \ldots, C_c\}$, in which case we will call it a $C$-matrix.
    Let $k$ be the maximal number of $B$-matrices that appear as factors in any of the summands. Then $k\geq \ell$. 
    Consider any summand with the maximal number of $B$-matrices in its factors. Suppose it is the first summand. %
    For simplicity of notation, we assume the first $k$ factors of this summand are $B$-matrices and the last $m - k$ factors are $C$-matrices: %
    \[
    D_{i_{1,1}} \otimes \cdots \otimes D_{i_{1,m}}  = B_{i_{1,1}}\otimes\cdots\otimes B_{i_{1,k}}\otimes C_{i_{1,k+1}}\otimes\cdots\otimes C_{i_{1,m}}.
    \]
    We restrict the first $k$ factors of each summand in the definition of $A$ to $I\times J$ to get 
    \[
    A' = \sum_{t=1}^r a_t \, (D_{i_{t, 1}})_{I\times J} \otimes \cdots \otimes (D_{i_{t, k}})_{I\times J} \otimes D_{i_{t, k+1}} \otimes \cdots \otimes D_{i_{t, m}}.
    \]
    If there is a $C$-matrix among $D_{i_{t, 1}}, \ldots, D_{i_{t, k}}$, then $(D_{i_{t, 1}})_{I\times J} \otimes \cdots \otimes (D_{i_{t, k}})_{I\times J}$ is zero. 
    Moreover, nonzero summands cannot have any $B$-matrix among $D_{i_{t, k+1}}, \ldots, D_{i_{t, m}}$, as this would contradict the maximality of $k$. So in any non-zero summand the $D_{i_{t, 1}}, \ldots, D_{i_{t, k}}$ are $B$-matrices and the $D_{i_{t, k+1}}, \ldots, D_{i_{t, m}}$ are $C$-matrices. 
    We know that the first $k$ factors of the first summand of $A$ are $B$-matrices, so $A'$ is nonzero by the linear-independence of $\{(B_i)_{I\times J}\}^{\otimes k}\otimes\{C_i\}^{\otimes m-k}$.
    Using the notation $B'_i = (B_i)_{I\times J}$, we may rewrite $A'$ as
    %
    \[
    A' = \sum_{{j}\in[b]^k}B'_{j_1}\otimes\cdots\otimes B'_{j_k}\otimes M_j,
    \]
    where
    \[
    M_j = \sum_{t=1}^r\delta_{i_{t,1}=j_1}\cdots\delta_{i_{t,k}=j_k}\cdot a_t\, C_{i_{t,k+1}}\otimes\cdots\otimes C_{i_{t,m}}.
    \]
    Note that $S = M_{(i_{1,1},\dots,i_{1,k})}$ is a nonzero matrix, because 
    $a_1\neq0$ and the set $\{C_i\}_{i\in\{b+1,\dots,c\}}^{\otimes(m-k)}$ is linearly independent. %
    %
    %
    %
    %
%
    %
%
    Suppose $S_{x,y}\neq 0$.
    Then
    \[
    \sum_{{j}\in[b]^k}B'_{j_1}\otimes\cdots\otimes B'_{j_k} \cdot (M_j)_{x,y}
    \]
    is a nontrivial linear combination, so it has rank at least 
    \[
    \minrnk(X_{I\times J})^k\geq \minrnk(X_{I\times J})^\ell.
    \]
    This matrix is a submatrix of $A'$, and $A'$ is a submatrix of $A$. This proves the claim.
    %
\end{proof}

\subsection{Narrow tensors have maximal asymptotic subrank}
\label{sec:narrow}\label{subsec:main-unb}

We will now %
prove the main theorem of this section, which says that for every constant~$c$, and large enough $n_1, n_2$, all tensors in $\FF^{n_1} \otimes \FF^{n_2} \otimes \FF^{c}$ have maximal asymptotic subrank.

\begin{theorem}%
\label{theorem:concise-cn1n2-assubrank-c}
    For every $c\geq2$ there is $N(c)\in\mathbb{N}$ such that for any $n_1,n_2\in\NN$ such that $\max\{n_1,n_2\}>N(c)$, every concise tensor $T\in\mathbb{F}^{n_1}\otimes\mathbb{F}^{n_2}\otimes\mathbb{F}^c$ satisfies $\asympsubrank(T)=c$.
\end{theorem}






%


\begin{proof}
Suppose $n_1 \geq n_2$. Suppose $n_1 > N(c)$ for an $N(c)$ to be determined in the proof.
Let $T \in \FF^{n_1} \otimes \FF^{n_2} \otimes \FF^{c}$ be concise.
Let $\mathcal{A} \subseteq \FF^{n_1 \times n_2}$ be the 3-slice span of $T$.
By \autoref{cla:high-rk} we have 
$\maxrank(\mathcal{A}) \geq \frac{1}{c} n_1$.
We apply \autoref{lem:minrk_diag} to $\mathcal{A}$, giving matrices $B_1, \ldots, B_b$, $B_{b+1}, \ldots, B_c$ and a subset $J \subseteq [\min_\ell n_\ell]$.
To these, we apply \autoref{claim:linear-combinations-of-mixed-tensorprods-high-rank} for $m \geq \ell \geq 1$ (to be specified later) 
to obtain a set~$Y$ of slices from~$T^{\boxtimes m}$
satisfying
\beq\label{eq:mrkY}
\minrank(Y) \geq \big(\eps(c)\maxrank(\mathcal{A})\big)^\ell \geq (\eps(c)\tfrac1c n_1)^{\ell}. 
\eeq
The size of~$Y$ equals the number of $m$-tuples in $[c]^m$ containing at least~$\ell$ elements from~$[b]$.
Taking $m \geq 8c$ and $\ell = \tfrac{1}{2c} m$, it then follows from the Chernoff bound (e.g., \cite[Equation~(7)]{MR1045520}) that~$|Y|\geq \tfrac12c^m$.

Since~$T$ is concise, so is~$T^{\boxtimes m}$ and therefore, the set~$Y$ is linearly independent.
Let~$T'$ be the tensor with 3-slices given by the elements in~$Y$.
Then $T^{\boxtimes m} \geq T'$.
Provided
\begin{equation}\label{eq:minrank}
\minrnk(Y) \geq 2(\tfrac12c^{m})^2,
\end{equation}
it follows from \autoref{cla-high-rank-high-subrank} that
\beq\label{eq:QTmbound}
\subrank(T^{\boxtimes m}) \geq \subrank(T') \geq \tfrac12c^m.
\eeq
By our choice of~$\ell$ and~\eqref{eq:mrkY}, we have that~\eqref{eq:minrank} holds if
\beqn
\big(\eps(c)\tfrac1c n_1\big)^{\tfrac{1}{2c}m}
\geq
2(\tfrac12c^{m})^2.
\eeqn
It is clear that there is an~$N(c)\in\NN$ such that this is satisfied for any~$n_1>N(c)$ and every~$m$.
Letting~$m$ go to infinity we then get from~\eqref{eq:QTmbound} that 
\beqn
\asympsubrank(T) = \lim_{m\to \infty}\subrank(T^{\boxtimes m})^{1/m} \geq c.\qedhere
\eeqn
\end{proof}

\begin{remark}
    To conclude this section, let us note that a similar approach fails for $k$-tensors, for $k>3$. The motivating question mentioned at the beginning of \autoref{sec:lower-bound-imbalanced}, whether there exist accumulation points in the image of asymptotic slicerank (over~$\CC$ and finite fields), makes sense for $k$-tensors for any $k$. In order to show there cannot exist such points in the same way for $k>3$ we would need to show there are only finitely many formats that allow asymptotic sliceranks strictly smaller than the maximal constant $c$. For $k=3$ we showed it for asymptotic subrank, which lower bounds asymptotic slicerank, but this is false for asymptotic subrank and any~$k>3$. For $k=3$ conciseness implies the formats that allow for a constant asymptotic subrank are only $(n_1,n_2,c)$ for a constant $c$ and $n_1,n_2$ growing. But for $k=4$ and above there can be formats of the form $(n_1,n_2,\dots,n_{k-2},c_1,c_2)$ for $c_1,c_2$ constants and $n_1,n_2,\dots,n_{k-2}$ growing. For these formats the straightforward generalisation 
    of \autoref{theorem:concise-cn1n2-assubrank-c}, asking to show that for every $c_1,c_2$ constants there would be $N$ for which any concise $(n_1,n_2,\dots,n_{k-2},c_1,c_2)$-tensor with $n_1,n_2,\dots,n_{k-2}>N$ has asymptotic subrank $\min\{c_1,c_2\}$ does not hold for $k$-tensors, for any $k>3$:
    
    The $4$-tensor $\id_n\otimes \id_2=\sum_{i=1}^n e_i\otimes e_i\otimes e_1\otimes e_1+\sum_{i=1}^n e_i\otimes e_i\otimes e_2\otimes e_2$ is a concise $4$-tensor of format $(n,n,2,2)$ (for each $i\in[4]$, the set of $i$-slices is linearly independent) but it is a tensor-product as a $2$-tensor (for the bipartition $\{\{1,2\},\{3,4\}\}$) so its subrank and its asymptotic subrank are $1$ as a $2$-tensor and thus as a $4$-tensor as well.
    
    This example can be generalised to any $k>3$ and any $c_1=c_2=:c$ as $\langle n\rangle_{k-2}\otimes \id_c$ where $\langle n\rangle_{k-2}=\sum_{i=1}^n e_i\otimes\cdots\otimes e_i$ is the diagonal (unit) $(k-2)$-tensor of size $n$.
\end{remark}

\subsection{Concrete case: smallest dimension equals two} %
\label{subsec:dimtwo}

In the previous section we proved for any constant $c$ that for large enough $n_1,n_2$ all concise tensors in $\FF^{n_1} \otimes \FF^{n_2} \otimes \FF^c$ have asymptotic subrank $c$. In the special case $c=2$ we can prove the following stronger statement in which we make the constant $N(2)$ explicit, namely $N(2) = 2$, and we are able to replace asymptotic subrank by subrank.

\begin{theorem}\label{local-dim-2}
    Let $n_1,n_2 > 2$ and let $T\in\mathbb{F}^{n_1}\otimes\mathbb{F}^{n_2}\otimes\mathbb{F}^2$ be a concise tensor. Then
    \[\subrank(T)=2.\]
\end{theorem}

\begin{example}
    In \autoref{local-dim-2} it is required that at least one of the dimensions~$n_1, n_2$ is strictly larger than 2.
    Namely, there are concise tensors $T \in \FF^2 \otimes \FF^2 \otimes \FF^2$ such that $\subrank(T) = 1$, for instance the tensor $T=e_1 \otimes e_2 \otimes e_2 + e_2 \otimes e_1 \otimes e_2 + e_2 \otimes e_2 \otimes e_1$ that is often called the ``W-tensor''. 
\end{example}

\begin{remark}
    \autoref{local-dim-2} leaves open what happens for $n_1=2$ and $n_2 > 2$. 
    Over the complex numbers the classification %
    of Miyake and Verstraete \cite{miake2004multipartite22n} shows that indeed for every $n>2$ and every concise tensor $T \in \FF^{n} \otimes \FF^2 \otimes \FF^2$, $\subrank(T)=2$.
    Over any algebraically closed field, Tobler \cite[Section 4.2]{Tobler1997Spezi-6235} showed that for every $n_2>n_1\geq2$ and every concise tensor $T \in \FF^{n_1} \otimes \FF^{n_2} \otimes \FF^2$, $\bordersubrank(T)=2$. This in particular implies that the asymptotic subrank is $2$ (\autoref{lem:border-asymp-subrank}). %

\end{remark}

\begin{lemma}\label{cl:1000andtus1Sub2} 
For any $s,t,u \in \FF$, the tensor with slices 
\[
A=\begin{pmatrix}1 & 0\\ 0 & 0\end{pmatrix}, B=\begin{pmatrix}t & u\\s & 1\end{pmatrix}
\]
has subrank $2$.
\end{lemma}
\begin{proof} We first subtract $s$ times the second column from the first column and $u$ times the second row from the first row to get 
\[
A=\begin{pmatrix}1 & 0\\ 0 & 0\end{pmatrix}, B=\begin{pmatrix}t' & 0\\ 0 & 1\end{pmatrix}
\]
for some $t'$.
We then subtract $t'\cdot A$ from $B$ to get 
\[
A=\begin{pmatrix}1 & 0\\ 0 & 0\end{pmatrix}, B=\begin{pmatrix}0 & 0\\ 0 & 1\end{pmatrix}.
\]
This is $\langle2\rangle$, the diagonal tensor of rank $2$.
\end{proof}

\begin{proof}[Proof of \autoref{local-dim-2}]
    Let $A,B$ be the two $n_1\times n_2$ slices. We diagonalize the $n_1\times 2n_2$ matrix $(A|B)$ by diagonalizing $A$ and doing the same row and column operations on~$B$. Observe that $\rank(A|B)=n_1$ as $T$ is concise. 
    Assume w.l.o.g.\ that the top $r\times r$ minor of $A$ is the identity and the rest is~$0$. Specifically, the last $n_1-r$ rows in it are $0$. %
    We can now continue and diagonalize the last $n_1-r$ rows of $B$, noting that this will not change the structure of~$A$ (we are only doing row operations now).
    We will get a subset~$I$ of $B$'s columns such that restricted to some of the last $n_1-r$ rows, $B|_{\{r+1,\dots,r+\ell\}\times I}$ is the identity. 
    By doing more row operations we can assume that the first $r$ rows in the columns~$I$ in $B$ are $0$ (the $n_1-r$ rows in $A$ are zeros so we are free to subtract scaled versions of these rows from above rows to zero out the desired cells in $B$). Here's an illustration of our two $n_1\times n_2$ matrices (the dashed lines separate the first~$r$ rows and columns from the rest and the $I$ indices of columns are indicated beneath matrix $B$. Asterisks signify arbitrary values that we cannot assume anything about). In particular, a $0$ above the dashed line in $B$ represents an $r\times 1$ columns of zeros:
    \[
   A=\left(
   \begin{array}{c;{2pt/2pt}c}
       \id_r & 0 \\ \hdashline[2pt/2pt]
       0 & 0
   \end{array}
   \right) \quad
   B=\genfrac{}{}{0pt}{}{\left(
   \begin{array}{cccccc;{2pt/2pt}c}
        *&0&*&0&\cdots & 0& \\ \hdashline[2pt/2pt]
        *&1&*&0&\cdots&0&*\\
        *&0&*&1&\cdots&0&*\\
        *&\vdots&*&\vdots&\ddots&\vdots&\vdots\\
        *&0&*&0&\cdots&1&*\\
        *&0&*&0&\cdots&0&*
   \end{array}
   \right)}{ \begin{array}{ccccccc}
        &I_1&&I_2&\cdots&I_\ell&\hspace{1pt}  
   \end{array}}%
    \]
    We now have three cases to consider:

    Case 1: $1<r<\min\{n_1,n_2\}$.
    Let $i\in I$ be arbitrary but different than $r$.
    Note that $B_{r+i,i}=1$ (it is part of the identity matrix on columns $I$). 
    Consider the restriction of our tensor to rows $r,r+i$ and columns $r$ and $i$. 
    This gives the tensor with slices
    \[
    A = \begin{pmatrix}1 & 0\\ 0 & 0\end{pmatrix},
    B=\begin{pmatrix}t & 0\\ s & 1\end{pmatrix}
    \]
    for some $s$ and $t$. From \autoref{cl:1000andtus1Sub2} we know that this tensor has subrank 2.

    Case 2: $r=\min\{n_1,n_2\}$.
We consider two subcases:
    \begin{itemize}
    \item Case 2a: The matrix $B$ has a nonzero nondiagonal coefficient. Assume w.l.o.g.\ that $B_{1,2}=t \neq 0$.  
Restrict our tensor to rows 1,3 and columns 2,3 ($n_1,n_2\geq3$).
We get slices
\[
A=\begin{pmatrix}0 & 0\\ 0 & 1\end{pmatrix}, B=\begin{pmatrix}t & *\\ * & *\end{pmatrix}.
\]
Using \autoref{cl:1000andtus1Sub2}, as $t\neq 0$, this tensor has subrank $2$.

    \item Case 2b: All nondiagonal coefficients of $B$ are $0$ ($\forall i\neq j: B_{i,j}=0$). 
As the tensor is concise, the matrix~$B$ is not a multiple of the matrix $A$. The diagonal coefficients of $A$ are all~$1$. Hence, w.l.o.g., $B_{1,1}\neq B_{2,2}$.
Restricting to rows and columns $\{1,2\}$ we get 
\[
A=\begin{pmatrix}1 & 0\\ 0 & 1\end{pmatrix},
B=\begin{pmatrix}a & 0\\ 0 & b\end{pmatrix}.
\]
Taking appropriate linear combinations of $A$ and $B$ we get 
\[
A=\begin{pmatrix}0 & 0\\ 0 & 1\end{pmatrix}, B=\begin{pmatrix}1 & 0\\ 0 & 0\end{pmatrix}.
\]
\end{itemize}

    Case 3: $r=1$.
    In this case we repeat the argument above with $B$ in place of $A$. As $\rank(A)=1$, $n_1,n_2>2$, and the tensor is concise we must have $\rank(B)>1$.
\end{proof}

\begin{remark}
    The proof of \autoref{local-dim-2} only uses conciseness with respect to two directions, namely that the matrices $A$ and $B$ are linearly independent (Case 2b) and that the matrix $(A|B)$ has full rank (all cases). An example (that is not concise) is
    \[
    A=\begin{pmatrix}1 & 1\\ 1 & 1\end{pmatrix}, B=\begin{pmatrix}1 & 1\\ 2 & 2\end{pmatrix}.
    \]
\end{remark}

\section{Discreteness of asymptotic tensor ranks}
\label{sec:discreteness}

In this section we prove a general sufficient condition for a tensor parameter to be discrete (using the lower bounds on the asymptotic subrank of concise tensors that we proved in the previous sections, \autoref{sec:max-rank} %
and \autoref{sec:lower-bound-imbalanced}). Then we show that this condition is satisfied (and thus obtained discreteness) for several natural tensor parameters and regimes.

\subsection{General sufficient condition for discreteness}

We recall that two tensors $S$ and $T$ are called equivalent if $S \geq T$ and $T \geq S$. Every tensor is equivalent to a concise tensor (\autoref{lem:equiv-concise}). We will be interested in functions~$f$ on tensors that are invariant under equivalence, meaning that if two tensors~$S$ and $T$ are equivalent, then $f(S) = f(T)$. This is a very mild condition, in particular any function that is monotone under restriction is also invariant under equivalence.

\begin{theorem}
    \label{th:gen-no-acc}
    Let $\FF$ be any field. Let $\mathcal{C}$ be any subset of order-three tensors over~$\FF$, such that for any $T \in \mathcal{C}$ there is an $S \in \mathcal{C}$ that is concise and equivalent to $T$.
    Let~$f : \mathcal{C} \to \RR_{\geq0}$ be any function such that
    \begin{enumerate}[\upshape (i)]
    \item\label{item:equiv} $f$ does not change under equivalence of tensors
    \item\label{item:lb} $f(T) \geq \asympsubrank(T)$ for every $T \in \mathcal{C}$
    \item\label{item:fin} For every $n_1, n_2, n_3\in \NN$,  $f$ takes finitely many values on $\mathcal{C} \cap (\FF^{n_1} \otimes \FF^{n_2} \otimes \FF^{n_3})$ 
    \item\label{item:ub} For every $n_1, n_2, n_3\in \NN$, $f(T) \leq \min_\ell n_\ell$ for every $T \in \mathcal{C} \cap (\FF^{n_1} \otimes \FF^{n_2} \otimes \FF^{n_3})$.
    \end{enumerate}
    Then the set of values that $f$ takes on $\mathcal{C}$ has no accumulation points.
\end{theorem}

\begin{proof}
    Suppose we have, for every $i \in \NN$, three numbers $a_i, b_i, c_i \in \NN$ and a tensor
    $T_i \in \mathcal{C} \cap (\FF^{a_i} \otimes \FF^{b_i} \otimes \FF^{c_i})$ in such a way that the set $\{f(T_i): i \in \NN\}$ is infinite. We will show that the $f(T_i)$ cannot converge.
    
    We may assume that every tensor $T_i$ is concise by condition~(\ref{item:equiv}) and \autoref{lem:equiv-concise}.
    By condition~(\ref{item:fin}), $f$ takes only finitely many values on $\mathcal{C} \cap (\FF^{n_1} \otimes \FF^{n_2} \otimes \FF^{n_3})$ for any choice of $n_1, n_2, n_3$, so the set of triples $\{(a_i, b_i, c_i): i \in \NN\}$ must be infinite. %
    
    We consider two cases. First,  suppose $\left\{\min \{a_i, b_i, c_i\}\right\}_i$ is bounded. Suppose for simplicity of notation that $a_i = \min \{a_i, b_i, c_i\}$ for every $i \in \NN$. Then $a_i$ takes only finitely many values. Take any infinite subsequence of the $T_i$ such that  $a_i$ is constant. %
    Then both $b_i$ and $c_i$ are unbounded, because $b_i\leq a_i c_i$ and $c_i\leq a_i b_i$ by \autoref{cl:concise-formats}. When~$b_i$ and $c_i$ are large enough we have $\asympsubrank(T_i) = a_i$ by \autoref{theorem:concise-cn1n2-assubrank-c}, %
    and hence $f(T_i) = a_i$ by condition~(\ref{item:lb}) and~(\ref{item:ub}). We conclude that $f(T_i)$ takes only finitely many values, which contradicts our assumption that $\{f(T_i): i \in \NN\}$ is infinite.
    
    Second, suppose $\min \{a_i, b_i, c_i\}$ is unbounded. 
    We have $f(T_i) \geq \asympsubrank(T_i)$ by condition~(\ref{item:lb}). We have $\asympsubrank(T_i) \geq \min \{a_i, b_i, c_i\}^{1/3}$ by \autoref{th:cube-root-lb}. It then follows that the values $f(T_i)$ do not converge.
\end{proof}

\subsection{Asymptotic subrank and asymptotic slice rank}

We will now use \autoref{th:gen-no-acc} to prove a discreteness theorem for asymptotic slice rank and asymptotic subrank. 

Before stating our theorem, we recall the notions of tight and oblique tensors, which are used in one part of the theorem. 
Let $T \in \FF^{n_1} \otimes \FF^{n_2} \otimes \FF^{n_3}$.
We recall that the support of $T$ is the set $\supp(T) = \{(i,j,k) : T_{i,j,k} \neq 0\}$.
The tensor $T$ is called \emph{tight} if (after an appropriate basis transformation) there exist injective maps $f_\ell : [n_\ell] \to \ZZ$ such that for every $(i_1,i_2,i_3) \in \supp(T)$ we have $f_1(i_1) + f_2(i_2) + f_3(i_3) = 0$. 
The tensor $T$ is called \emph{oblique} if (after an appropriate basis transformation) $\supp(T)$ is an antichain in the product order (i.e.~pointwise order) on $[n_1] \times [n_2] \times [n_3]$. Tight tensors are oblique.

\begin{theorem}
\label{th:no-acc}
\hfill
\begin{enumerate}[\upshape(i)]
\item Let $\mathbb{F}$ be any field and $S \subseteq \FF$ any finite subset. 
Then the asymptotic subrank and asymptotic slice rank on order-three tensors over $S$ have no accumulation points.
\item 
The asymptotic slice rank on order-three tensors over~$\CC$ has no accumulation points.
\item Let $\mathbb{F}$ be any field.
\begin{itemize}
\item The asymptotic subrank and asymptotic slice rank on tight order-three tensors over $\mathbb{F}$ have no accumulation points.
\item The asymptotic slice rank on oblique order-three tensors over $\mathbb{F}$ has no accumulation points.
\end{itemize}
\end{enumerate}
\end{theorem}

\begin{proof}
    The proof of each case follows from \autoref{th:gen-no-acc}. As slice rank is lower bounded by subrank, $f(T) \geq \asympsubrank(T)$ for every $T \in \mathcal{C}$, over any field $\FF$. Thus, conditions (i), (ii) and (iv) of  \autoref{th:gen-no-acc} are fulfilled. We next show that condition (iii) of \autoref{th:gen-no-acc} holds for each of the cases in \autoref{th:no-acc}, i.e. that in each of these cases $f$ takes only finitely many values on $\mathcal{C} \cap (\FF^{n_1} \otimes \FF^{n_2} \otimes \FF^{n_3})$.

    (i) Let $f$ be the asymptotic subrank or the asymptotic slice rank and let $\mathcal{C}$ be the set of all order-three tensors over the finite set $S$. Indeed $\mathcal{C}$ is closed under taking the concise version of any element by \autoref{lem:equiv-concise}, since any tensor is equivalent to a concise subtensor (which in particular again has coefficients in $S$). Then $f$ takes only finitely many values on $\mathcal{C} \cap (\FF^{n_1} \otimes \FF^{n_2} \otimes \FF^{n_3})$ simply because the set of tensors in $\FF^{n_1} \otimes \FF^{n_2} \otimes \FF^{n_3}$ with coefficients in $S$ is finite. 
    
    (ii) Let $f$ be the asymptotic slice rank and let~$\mathcal{C}$ be the set of all order-three tensors over~$\CC$.  Fix dimensions $n_1,n_2$ and $n_3$. By the characterization of asymptotic slice rank in terms of entanglement (moment) polytopes (via the quantum functionals) \cite[Corollary 5.7]{MR4495838} and the fact that for fixed $n_1,n_2,n_3$ there are finitely many such polytopes \cite[Theorem 2]{walter2013entanglement}, it holds that $f$ takes only finitely many values on $\mathcal{C} \cap (\FF^{n_1} \otimes \FF^{n_2} \otimes \FF^{n_3})$. 
    
    (iii) Let $f$ be the asymptotic subrank or the asymptotic slice rank and let $\mathcal{C}$ be the set of all tight order-three tensors over arbitrary $\FF$. Fix dimensions $n_1,n_2$ and $n_3$. By the characterization of asymptotic subrank in terms of the tight support (via the support functionals) \cite[Proposition 5.4]{strassen1991degeneration} and \cite[Theorem 5.9]{MR4495838}, and the fact there are only a finite number of possible supports, it holds that $f$ takes only finitely many values on $\mathcal{C} \cap (\FF^{n_1} \otimes \FF^{n_2} \otimes \FF^{n_3})$.
    

    Let $f$ be the asymptotic slice rank and let $\mathcal{C}$ be the set of all oblique order-three tensors over arbitrary $\FF$. By the characterization of asymptotic slice rank in terms of the oblique support \cite[Proposition 4]{sawin} it follows that $f$ takes only finitely many values on $\mathcal{C} \cap (\FF^{n_1} \otimes \FF^{n_2} \otimes \FF^{n_3})$. 
\end{proof}

\begin{remark}
    \autoref{th:no-acc} gives several regimes in which the asymptotic subrank and asymptotic slice rank are discrete, that is, between every two consecutive values there is a ``gap''.
    Using \autoref{ex:strassenNullAlgebra} we can say something about the size of these gaps. Namely,
    \autoref{ex:strassenNullAlgebra} gives for every large enough $n$ (namely $n\geq5$) a concise tensor $T \in \FF^n \otimes \FF^n \otimes \FF^n$ such that the asymptotic subrank and asymptotic slice rank of $T$ equal $2\sqrt{n-1}$. This means that the gap between two consecutive values is at most $2\sqrt{n-1} - 2\sqrt{n-2}$, which is at most $O(1/\sqrt{n})$.
    This shows in particular that the consecutive gaps of the asymptotic subrank and asymptotic slice rank tend to zero as $n$ grows.
\end{remark}

As sketched in the introduction, with similar ideas as above (but much simpler) we can prove the analogous statement for asymptotic rank over finite fields.
Recall that the asymptotic rank of a tensor~$T$ is defined as $\asymprank(T) = \lim_{n\to\infty} \tensorrank(T^{\boxtimes n})^{1/n}$.

\begin{theorem}\label{th:asympran-disc}
    Let $\FF$ be any field, $S\subseteq\FF$ a finite set, and $k \in \mathbb{N}$. The asymptotic rank on $k$-tensors over~$S$ has no accumulation points. 
    In particular, the asymptotic rank on $k$-tensors over any finite field has no accumulation points.
\end{theorem}
\begin{proof} 
We may assume all our tensors to be concise (\autoref{lem:equiv-concise}).
For a fixed format $(n_1, \ldots, n_k)$, there are only finitely many possible values of the asymptotic rank of tensors in $\FF^{n_1} \otimes \cdots \otimes \FF^{n_k}$ with coefficients from $S$ (simply because there are only finitely many tensors of that format with coefficients from $S$). Thus, to make a sequence of concise tensors $T_i$ such that the asymptotic rank $\asymprank(T_i)$ takes infinitely many different values, we need infinitely many different formats $(n_1,\ldots, n_k)$. In particular, $\max_\ell n_\ell$ needs to be unbounded. Since $\asymprank(T_i) \geq \max_\ell n_\ell$, the sequence $\asymprank(T_i)$ cannot converge.
\end{proof}

\section{Pivots of tensors and better lower bounds on the asymptotic subrank}\label{sec:lower-bound-asympsubrank}

In \autoref{sec:max-rank} we proved a cube-root lower bound on the asymptotic subrank of concise tensors, using a result on max-rank and matrix multiplication tensors.

In this section we develop a different strategy, using pivots, to first reprove the cube-root lower bound on asymptotic subrank of \autoref{sec:max-rank} and then prove a better lower bound under a symmetry condition that we introduce (pivot-matched).

Namely we show that the asymptotic subrank of a concise pivot-matched tensor is at least the square-root of the smallest dimension. In particular, this holds for symmetric concise tensors. %

\subsection{Pivots of tensors}
We will first discuss what a pivot of a matrix is and what the pivots of a matrix subspace are, then introduce the pivots of a tensor, then relate this notion to the border subrank and the asymptotic subrank.

\subsubsection{Pivot of a matrix and pivots of matrix subspaces}

Following Meshulam \cite{meshulam1985maximal}, we recall the notion of pivots of a matrix subspace and associated covering and packing parameters $\rho$ and $\sigma$, which are known to be equal by an application of Kőnig's theorem. In the next section we will use these notions in the context of tensors.

\begin{definition}
    For any $n_1 \times n_2$ matrix $A$ we call the smallest element $(i,j) \in [n_1] \times [n_2]$ in the lexicographic ordering on $[n_1] \times [n_2]$ for which $A_{i,j} \neq 0$ the \emph{pivot} of $A$, and we denote this element by $p(A)=(f(A),g(A))$. 
\end{definition}

\begin{lemma}\label{lem:pivot-set-invariant}
    Let $\mathcal{A}$ be a subspace of $n_1\times n_2$ matrices. There is a basis for $\mathcal{A}$, $\{A_1,\dots,A_{n_3}\}$, with distinct pivots $\forall i\neq j: p(A_i)\neq p(A_j)$.
    Additionally, for any such basis, the set of its matrices' pivots $\{p(A_1),\dots,p(A_{n_3})\}$ is the same.
\end{lemma}
To see this, vectorise the matrices row-by-row and stack them as the rows of a matrix $M$. Changing the basis amounts to invertible row operations on $M$ and having distinct pivots amounts to getting $M$ into its row-echelon form (up to the order of the rows). Uniqueness of the row-echelon form gives the uniqueness of the set of pivots.

By this lemma we may define the set of pivots as a parameter of the linear subspace of matrices itself:

\begin{definition}\label{def:pivots-space}
    For any subspace of $n_1 \times n_2$ matrices $\mathcal{A}$, let $A_1, \ldots, A_{n_3}$ be any basis of $\mathcal{A}$ such that the pivots of the $A_i$ are distinct. We call the set of all pivots $\{p(A_1), \ldots, p(A_{n_3})\}$ the \emph{pivots} of $\mathcal{A}$, and we denote this set by $p(\mathcal{A})$. 
\end{definition}

\begin{remark}\label{rem:pivots-as-cell-locations}
    Let $T$ be a tensor with distinct pivots of its $1$-slices, then this set of pivots is a special case of a linearly independent set of cell-locations for $T$ in direction~$1$.
\end{remark}

\begin{definition}\label{def:rho}
    For any subspace of $n_1 \times n_2$ matrices $\mathcal{A}$, we let $\rho(\mathcal{A})$ be the smallest number $r$ such that there are subsets $S_1 \subseteq [n_1], S_2 \subseteq [n_2]$ of size $|S_1| + |S_2| = r$ for which $p(\mathcal{A}) \subseteq (S_1 \times [n_2]) \cup ([n_1] \times S_2)$.
\end{definition}

\begin{definition}
    For any subspace of $n_1 \times n_2$ matrices $\mathcal{A}$, we let $\sigma(\mathcal{A})$ be the largest number $s$ such that there are matrices $A_1, \ldots, A_s \in \mathcal{A}$ for which for every $i\neq j$ we have $f(A_i) \neq f(A_j)$ and $g(A_i) \neq g(A_j)$.%
\end{definition}

The following is an application of Kőnig's theorem on the equivalence of maximum matchings and minimal vertex covers in bipartite graphs 

\begin{lemma}[\cite{meshulam1985maximal}] \label{lem:pivots-max-ind-min-lines}
    For any subspace of $n_1 \times n_2$ matrices $\mathcal{A}$ we have $\rho(\mathcal{A}) = \sigma(\mathcal{A})$.
\end{lemma}

\subsubsection{Pivots of a tensor}

Having discussed the notion of pivots of a matrix subspace, we will now associate to any tensor matrix subspaces obtained by taking the span of the slices of the tensor. We have considered slice spans before when defining the parameters $\subrank_\ell$, but here we will need to be more precise and distinguish six matrix subspaces instead of three (because the property we consider is not invariant under taking transpose). We then consider the pivots of these subspaces. We will relate these to the parameters $\subrank_\ell$ and to the border subrank. Combining these ideas we then give another proof of the cube-root lower bound on the asymptotic subrank.

\begin{definition}\label{def:pivots-tensors}
    Let $T = \sum_{i,j,k} T_{i,j,k}\, e_i \otimes e_j \otimes e_k \in \FF^{n_1} \otimes \FF^{n_2} \otimes \FF^{n_3}$ be a tensor.
    Define the subspaces
    \begin{align*}
        &\mathcal{A}_{1,2} = \linspan \{A_1, \ldots, A_{n_3}\} \subseteq \FF^{n_1} \otimes \FF^{n_2},\quad A_k = \sum_{i,j} T_{i,j,k}\, e_i \otimes e_j\\
        &\mathcal{A}_{2,1} = \linspan \{A_1, \ldots, A_{n_3}\} \subseteq \FF^{n_2} \otimes \FF^{n_1},\quad A_k = \sum_{i,j} T_{i,j,k}\, e_j \otimes e_i\\
        &\mathcal{A}_{1,3} = \linspan \{A_1, \ldots, A_{n_2}\} \subseteq \FF^{n_1} \otimes \FF^{n_3},\quad A_j = \sum_{i,k} T_{i,j,k}\, e_i \otimes e_k\\
        &\mathcal{A}_{3,1} = \linspan \{A_1, \ldots, A_{n_2}\} \subseteq \FF^{n_3} \otimes \FF^{n_1},\quad A_j = \sum_{i,k} T_{i,j,k}\, e_k \otimes e_i\\
        &\mathcal{A}_{2,3} = \linspan \{A_1, \ldots, A_{n_1}\} \subseteq \FF^{n_2} \otimes \FF^{n_3},\quad A_i = \sum_{j,k} T_{i,j,k}\, e_j \otimes e_k\\
        &\mathcal{A}_{3,2} = \linspan \{A_1, \ldots, A_{n_1}\} \subseteq \FF^{n_3} \otimes \FF^{n_2},\quad A_i = \sum_{j,k} T_{i,j,k}\, e_k \otimes e_j.
    \end{align*}
    Let $\rho_{{\ell_1},{\ell_2}}(T) = \rho(\mathcal{A}_{{\ell_1},{\ell_2}})$.
\end{definition}

\begin{remark}
    Note that $\rho_{{\ell_1},{\ell_2}}(T)$ can depend on the order of ${\ell_1},{\ell_2}$, which defines which direction enumerates the rows and which the columns. This is a direct consequence of using the lexicographic order in the definition of the pivots. For example, the tensor~$T$ with %
    \[
    \mathcal{A}_{1,2} = \Bigl\{\begin{pmatrix}
        0&0&1\\
        1&0&0
    \end{pmatrix},\begin{pmatrix}
        0&1&0\\
        0&0&0
    \end{pmatrix}\Bigr\},
    \]
    has 
    \[
    \mathcal{A}_{2,1} = \Bigl\{\begin{pmatrix}
        0&1\\
        0&0\\
        1&0
    \end{pmatrix},\begin{pmatrix}
        0&0\\
        1&0\\
        0&0
    \end{pmatrix}\Bigr\}.
    \]
    Then $\rho_{1,2}(T)=1\neq2=\rho_{2,1}(T)$. %
\end{remark}

\subsection{Alternative proof of cube-root lower bound via pivots} \label{subsec:cube-root-via-pivots}

In this section we will give another proof of the cube-root lower bound on asymptotic subrank (\autoref{th:cube-root-lb}) using pivots. %

\begin{theorem}\label{th:pivot-uncertainty}
    Let $T \in \FF^{n_1} \otimes \FF^{n_2} \otimes \FF^{n_3}$ be concise.
    For every distinct ${\ell_1},{\ell_2},{\ell_3}$ we have 
    \[
    \rho_{{\ell_1},{\ell_2}}(T) \max \{\subrank_{\ell_1}(T), \subrank_{\ell_2}(T)\} \geq n_{\ell_3}.
    \]
\end{theorem}
\begin{proof}
We give the proof for ${\ell_1}=2, {\ell_2}=3, {\ell_3}=1$ (the proof for the other cases is the same).
By an invertible linear transformation of direction $1$ (linear combinations of the $1$-slices) we may assume that for every $i\in[n_1]$, the $i$th $1$-slice is the only one with a non-zero entry in its pivot $p(i)=(f(i),g(i))$. Now we will prove that in every row there are at most $\subrank_2(T)$ pivots and in every column at most $\subrank_3(T)$. %
Let $x$ in the image of $f$, and project the second leg of $T$ to $e_x$. Restricted in the third leg to the columns with pivots in row $x$, $\{k\in[n_3]:\exists i\in[n_1], (f(i),g(i))=(x,k)\}$, this gives the following diagonal submatrix of the $x$th $2$-slice:
\[
\sum_{i: f(i) = x} e_i \otimes e_x \otimes e_{g(i)},
\]
which implies $\subrank_2(T) \geq | \{i : f(i) = x\}|$ for every $x$ in the image of $f$. Similarly, we have the inequality $\subrank_3(T) \geq | \{i : g(i) = y\}|$ for every $y$ in the image of $g$.
By \autoref{def:rho}, $p(\mathcal{A}_{2,3}) \subseteq (S_2 \times [n_3]) \cup ([n_2] \times S_3)$ for some subsets $S_\ell \subseteq [n_\ell]$ such that $|S_2| + |S_3| = \rho_{2,3}(T)$.
Both of these combine to:
\begin{align*}
n_1 &= |\{(f(i), g(i)) : i\in [n_1]\}|\\ &\leq |\{(f(i),g(i)): i \in [n_1], f(i) \in S_2\}| + |\{(f(i),g(i)) : i \in [n_1], g(i) \in S_3\}|\\
&\leq |S_2| \subrank_2(T) + |S_3| \subrank_3(T)\\
&\leq (|S_2| + |S_3|) \max \{\subrank_2(T), \subrank_3(T)\}\\
&= \rho_{2,3}(T) \max \{\subrank_2(T), \subrank_3(T)\}.
\end{align*}
This proves the claim.
\end{proof}

\begin{definition}
The \emph{border subrank} of a tensor $T \in \FF^{n_1} \otimes  \FF^{n_2} \otimes \FF^{n_3}$, denoted by~$\bordersubrank(T)$, is the largest number~$r$ such that there are matrices $A(\eps), B(\eps), C(\eps)$ whose coefficients are Laurent polynomials in the formal variable $\eps$, such that
\[
(A(\eps) \otimes B(\eps) \otimes C(\eps))  T = \langle r\rangle + \eps S_1 + \eps^2 S_2 + \cdots + \eps^t S_t
\]
where $\langle r\rangle$ is the $r \times r \times r$ diagonal tensor and the $S_i$ are arbitrary tensors (with coefficients in $\FF$). 
\end{definition}

The following known relation between border subrank and asymptotic subrank can be found in \cite[Proposition 5.10]{strassen1988asymptotic} and \cite[Proposition 15.26]{burgisser1997algebraic}.
\begin{lemma}\label{lem:border-asymp-subrank}
    For any tensor~$T$ we have $\bordersubrank(T) \leq \asympsubrank(T)$.
\end{lemma}

\begin{theorem}\label{th:rho-bordersubrank}
    For any tensor $T$ and any ${\ell_1}\neq {\ell_2}$ we have
    $\rho_{{\ell_1},{\ell_2}}(T) \leq \bordersubrank(T)$.
\end{theorem}
\begin{proof}
    Let $r = \rho_{{\ell_1},{\ell_2}}(T)$.
    Then by \autoref{lem:pivots-max-ind-min-lines}, after taking appropriate linear combinations of the slices of $T$, $T$ has an $r \times r \times r$ subtensor with slices $A_1, \ldots, A_r$ whose pivots $p(A_i)$ are pairwise disjoint in both coordinates. We may assume that $p(A_i) = (i,i)$ and that the $i$'th row of $A_i$ is zero except for the pivot coordinate $(i,i)$. Thus $A_i$ is of the form
    \[
    A_i = e_i \otimes e_i 
    + \sum_{j > i} \sum_{k} (A_i)_{j,k}\, e_j \otimes e_k.
    \]
    So $T$ restricts to the tensor
    \[
    S = \sum_i e_i \otimes e_i \otimes e_i 
    + \sum_{j > i} \sum_{k} (A_i)_{j,k}\, e_i \otimes e_j \otimes e_k.
    \]
    We define the matrices $A(\eps), B(\eps), C(\eps)$ as the ones that map
    \begin{align*}
    &A(\eps) : e_i \mapsto \eps^i e_i\\
    &B(\eps) : e_i \mapsto \eps^{-i} e_i\\
    &C(\eps) : e_i \mapsto e_i.
    \end{align*}
    Then 
    \[
(A(\eps) \otimes B(\eps) \otimes C(\eps))S = \sum_i e_i \otimes e_i \otimes e_i + \sum_{j > i} \eps^{j-i} \sum_{k} (A_i)_{j,k}\, e_i \otimes e_j \otimes e_k
    \]
    so $\bordersubrank(T) \geq r$.
\end{proof}

\begin{theorem} \label{thm:cube-root-via-pivots}
Let $T \in \FF^n \otimes \FF^n \otimes \FF^n$ be concise. Then
$\asympsubrank(T) \geq n^{1/3}$.
\end{theorem}

\begin{proof}
We consider two cases.
For the first case, suppose there are ${\ell_1},{\ell_2}$ such that $\rho_{{\ell_1},{\ell_2}}(T) \geq n^{1/3}$. Then $\asympsubrank(T) \geq n^{1/3}$, because generally we have $\asympsubrank(T) \geq \bordersubrank(T) \geq \rho_{{\ell_1},{\ell_2}}(T)$ (\autoref{th:rho-bordersubrank}). 
For the second case, suppose for all ${\ell_1} \neq {\ell_2}$ we have $\rho_{{\ell_1},{\ell_2}}(T) < n^{1/3}$. Then $\max \{ \subrank_{\ell_1}(T), \subrank_{\ell_2}(T) \} \geq n^{2/3}$ for all distinct ${\ell_1},{\ell_2}$ (\autoref{th:pivot-uncertainty}). From this follows that there are distinct ${\ell_1},{\ell_2}$ such that $\subrank_{\ell_1}(T) \geq n^{2/3}$ and  $\subrank_{\ell_2}(T) \geq n^{2/3}$. Then $\subrank(T^{\boxtimes 2}) \geq n^{2/3}$ (\autoref{lem:Qi-square}) so $\asympsubrank(T) \geq n^{1/3}$.
\end{proof}

We finish this subsection with a relation between the border subrank and the parameters~$\subrank_\ell$ that may be of independent interest.


\begin{theorem} \label{thm:commrnk_bounds_bordersubrnk} 
Let $T$ be any tensor.
Suppose $|\FF|=\infty$. 
For every $\ell \in [3]$ we have
\[\bordersubrank(T) \leq \subrank_\ell(T).\]
\end{theorem}

\begin{proof}
    We will give the proof for $\ell=3$.
    Let $q=\bordersubrank(T)$. Then there are matrices $A(\eps)$, $B(\eps)$, $C(\eps)$ whose coefficients are Laurent polynomials over $\FF$ in the variable $\eps$ %
    such that
    \[
    (A(\eps)\otimes B(\eps)\otimes C(\eps))(T)=\langle q\rangle+\eps S(\eps) %
    \]
    for some tensor $S(\eps)$ whose coefficients are polynomials over $\FF$ in $\eps$.
    Note that summing the 3-slices of $\langle q\rangle$ gives $\id_q$.
    Towards taking the sum of the 3-slices on both sides, define $C'(\eps) = (1,\ldots,1) C(\eps)$.
    Then
    $(A(\eps)\otimes B(\eps)\otimes C'(\eps))(T)=\id_q+\eps N(\eps)$ for some matrix $N(\eps)\in M_{q}(\FF[\eps])$. The determinant of the matrix $\id_q+\eps N(\eps)$ is a nonzero polynomial in $\eps$. 
    Thus, it has finitely many roots and we have another evaluation point for which it is non-zero. 
    Let $x\in\FF$ be such an evaluation point. Then $A(x),B(x),C'(x)$ are linear transformations over~$\FF$ for which we get a $3$-slice $(A(x)\otimes B(x)\otimes C'(x))(T)= \id_q +xN(x)\in M_q(\FF)$ of rank~$q$, showing $\subrank_3(T)\geq q=\bordersubrank(T)$.
\end{proof}

Combining \autoref{th:rho-bordersubrank} and \autoref{thm:commrnk_bounds_bordersubrnk}, assuming that $|\FF| = \infty$, 
we find for any directions ${\ell_1},{\ell_2},{\ell_3} \in [3]$ such that ${\ell_1}\neq {\ell_2}$ the relation
$\rho_{{\ell_1},{\ell_2}}(T) \leq \bordersubrank(T) \leq \subrank_{\ell_3}(T).$

\subsection{Square-root lower bound for symmetric tensors} \label{sub:sqrt-assub-equal-pivots}

For symmetric tensors we can improve the cube-root lower bound of \autoref{thm:cube-root-via-pivots} to a square-root lower bound, using the notion of pivots.

\begin{theorem}\label{cor:symmetric}
    Let $T \in \FF^{n} \otimes \FF^n \otimes\FF^n$ be symmetric and concise. Then  $\asympsubrank(T) \geq \sqrt{n}$.
\end{theorem}

In fact, we will prove a stronger theorem which implies \autoref{cor:symmetric}. For this we will define the notion of being pivot-matched.

First, for any collection of matrices $A_1, \ldots, A_n \in \FF^{m \times k}$ we define the pivot map 
\[
p = (f,g) : [n] \to [m] \times [k]
\]
by mapping $i\in [n]$ to the pivot $p(A_i) = (f(A_i), g(A_i))$ of $A_i$.

\begin{definition}[Pivot-matched] \label{def:pivot-matched}
    We call $T\in\FF^{n}\otimes\FF^{n}\otimes\FF^{n}$ \emph{pivot-matched} if there are tensors $T_A$ and $T_B$ isomorphic 
    to $T$ such that, 
    letting $A_i = \sum_{j,k} (T_A)_{i,j,k}\, e_j \otimes e_k \in \FF^{n} \otimes \FF^{n}$ for $i \in [n]$ be the 1-slices of $T_A$, and $p_A=(f_A,g_A)$ the pivot map of these matrices, and
    letting $B_k = \sum_{i,j} (T_B)_{i,j,k}\, e_i \otimes e_j \in \FF^{n}\otimes \FF^{n}$ for $k \in [n]$ be the 3-slices of $T_B$, and $p_B=(f_B,g_B)$ the pivot map of these matrices, we have that $p_A$ and $p_B$ are injective and  for all $i \in [n]$ we have $f_A(i)=f_B(i)$.
\end{definition}

The condition that $p_A$ and $p_B$ are injective can always be satisfied by \autoref{lem:pivot-set-invariant}.

We note that in the above definition, since pivots are defined using the lexicographic ordering, to make our proofs work it is crucial to define the slices $A_i$ as stated (and not as their transpose) and the same for the $B_i$.

\begin{remark}
    By reordering the $1$-slices of $T_A$ and the $3$-slices of $T_B$, the definition is equivalent to requiring a seemingly weaker requirement of having the same number of slices with pivots in each row for the two sets of slices, i.e.~$\forall i\in[n]:|f_A^{-1}(i)|=|f_B^{-1}(i)|$.
\end{remark}

\begin{theorem} \label{thm:samePivotsInRow_sqrtLowerBound_asymsub}
    Let $T\in\FF^{n}\otimes\FF^{n}\otimes\FF^{n}$ be a concise tensor that is pivot-matched. Then 
    \[\asympsubrank(T)\geq\sqrt{n}.\]
\end{theorem}

\begin{remark}
    We do not know whether \autoref{cor:symmetric} can be used to prove $\asympsubrank(T) \geq \sqrt{n}$ for general non-symmetric concise $T \in \FF^{n \times n \times n}$. 
    A natural approach is to symmetrize $T$ by taking $\mathrm{sym}(T) = \boxtimes_{\pi \in S_3} \pi T$ (where we let the symmetric group $S_3$ act by permuting the three tensor legs) and apply \autoref{cor:symmetric} to $\mathrm{sym}(T)$.
    We note, however, that the asymptotic subrank of a tensor generally does not behave nicely with respect to symmetrization (in contrast with matrix rank). 
    In particular, if we let $T = \langle n,1,1\rangle$, which has asymptotic subrank $1$, then the product of the cyclic permutations of $T$, $T \boxtimes (1,2,3) T \boxtimes (1,2,3)^2 T$, equals the matrix multiplication tensor $\langle n,n,n\rangle \in \FF^{n^2 \times n^2 \times n^2}$ which has asymptotic subrank~$n^2$. In other words, tensoring $T$ with permutations of $T$ can drastically increase the asymptotic subrank.
\end{remark}

\begin{lemma}\label{lem:upper-triangular}
    Suppose $T\in\FF^n\otimes\FF^n\otimes\FF^n$ is a concise tensor that is %
    pivot-matched. 
    Then
    $T\boxtimes T$ restricts to a tensor of the form
\begin{multline}\label{eq:req-form}
\sum_{i=1}^n \big[e_i\otimes e_{f(i)}\big]\otimes\big[e_{f(i)}\otimes e_{g(i)}\big] \otimes \big[e_{h(i)}\otimes e_i\big]
+ \sum_{i=1}^n\sum_{j<k}\sum_{l,u,v=1}^n c_{ijkluv}\big[e_i\otimes e_j\big]\otimes \big[e_k\otimes e_l\big]\otimes \big[e_u\otimes e_v].
\end{multline}
for $f,g,h:[n]\to[n]$ functions such that $i\mapsto(f(i),g(i))$ is injective.
\end{lemma}

\begin{proof}
Let $A_1,\dots,A_n$ be the $1$-slices of $T_A$ as in \autoref{def:pivot-matched}, and $B_1,\dots,B_n$ the $3$-slices of $T_B$. Without loss of generality, we may assume that
\beq\label{eq:pivot1}
(A_i)_{p_A(j)} = (B_i)_{p_B(j)} = \delta_{ij}
\eeq
where $p_A=(f_A,g_A)$ is the pivot map for the $1$-slices of $T_A$ (rows enumerated by direction~$2$) and $p_B=(f_B,g_B)$ is for the $3$-slices of $T_B$ (rows by direction $1$), otherwise we can take invertible linear combinations of the slices to guarantee it.

Using these two slicings, we have the decomposition
\beqn
T^{\boxtimes 2}\geq T_A\boxtimes T_B
=
\sum_{i,u=1}^n\sum_{j,v=1}^n\sum_{k,w=1}^n (A_i)_{jk}(B_w)_{uv} (e_i\otimes e_u)\otimes (e_j\otimes e_v)\otimes (e_k\otimes e_w).
\eeqn

We apply the following three linear maps to the respective legs of this tensor:
\begin{align*}
U &= \sum_{i=1}^n E_{i,i}\otimes E_{f_A(i),f_A(i)}\\
V &= \sum_{v=1}^n E_{f_B(v),f_B(v)}\otimes E_{g_B(v),g_B(v)}\\
W &= \sum_{w=1}^nE_{g_A(w),g_A(w)}\otimes E_{w,w}.
\end{align*}
where $E_{i,i}$ is the projector onto $e_i$.
A small calculation then gives
\begin{multline*}
T^{\boxtimes 2}\geq(U\otimes V\otimes W)(T_A\boxtimes T_B)
=\\
\sum_{i,v,w=1}^n
(A_i)_{f_B(v)\,g_A(w)}\, (B_w)_{f_A(i)\,g_B(v)}\, 
(e_i\otimes e_{f_A(i)})\otimes (e_{f_B(v)}\otimes e_{g_B(v)})\otimes (e_{g_A(w)}\otimes e_w).
\end{multline*}
This is the final tensor. We now proceed to show how it fits the desired form \eqref{eq:req-form}.

By definition of $f_A$ as giving the first non-zero rows of the $A$ slices, we get non-zero contributions only from the $(i,v)$-pairs satisfying $f_B(v) \geq f_A(i)$ from the first factor after the triple sum.
This already shows that there is an upper-triangular structure present, by looking at the second component of the first leg and the first component of the second leg.

To finish showing $(U\otimes V\otimes W)(T_A\boxtimes T_B)$ is in the desired form we focus on the pairs $(i,v)$ satisfying $f_A(i) = f_B(v)$, aiming to show that this fits the first sum in \eqref{eq:req-form}. 
Let~$S$ denote the set of such pairs.
Restricted to $S$, the triple sum becomes
\beqn
\sum_{(i,v)\in S}\sum_{w=1}^n 
(A_i)_{f_A(i)\,g_A(w)}\, (B_w)_{f_B(v)\,g_B(v)}\, 
(e_i\otimes e_{f_A(i)})\otimes (e_{f_B(v)}\otimes e_{g_B(v)})\otimes (e_{g_A(w)}\otimes e_w).
\eeqn
By~\eqref{eq:pivot1}, we have that 
\beqn
(B_w)_{f_B(v)\,g_B(v)} = (B_w)_{p_B(v)} = \delta_{vw}.
\eeqn
Hence, the sum over~$w$ collapses to a single term and we are left with the tensor 
\beqn
\sum_{(i,v)\in S}
(A_i)_{f_A(i)\,g_A(v)}\, 
(e_i\otimes e_{f_A(i)})\otimes (e_{f_B(v)}\otimes e_{g_B(v)})\otimes (e_{g_A(v)}\otimes e_v).
\eeqn
Fix $i,j$ such that $f_A(i) = j$.
Then~$v$ ranges over the set $f_B^{-1}(j)$, so all the $B$-slices whose pivot lies on row~$j$.
But by assumption, these are precisely the $A$-slices whose pivot lies on row~$j$.
This implies that the coordinate $(j,g_A(v))$ is the pivot of some $A$-slice.
Again by~\eqref{eq:pivot1} we  get that
\beqn
(A_i)_{f_A(i)\,g_A(v)} = \delta_{iv}.
\eeqn
Thus our tensor looks like
\beqn
\sum_{i=1}^n 
(e_i\otimes e_{f_A(i)})\otimes (e_{f_B(i)}\otimes e_{g_B(i)})\otimes (e_{g_A(i)}\otimes e_i).
\eeqn
Note that the maps $i \mapsto (i, f_A(i))$, $i \mapsto (f_B(i), g_B(i))$ and $i \to (g_A(i), i)$ are injective, as required
(i.e.\ this first sum is a diagonal of length~$n$).
\end{proof}

\begin{lemma} \label{lem:border-subrank-for-upper-triag}
Suppose that $T\in\FF^{n^2}\otimes\FF^{n^2}\otimes\FF^{n^2}$ restricts to a tensor $S$ of the form
\begin{multline*}
\sum_{i=1}^n \big[e_i\otimes e_{f(i)}\big]\otimes\big[e_{f(i)}\otimes e_{g(i)}\big] \otimes \big[e_{h(i)}\otimes e_i\big]
+ \sum_{i=1}^n\sum_{j<k}\sum_{l,u,v=1}^n c_{ijkluv}\big[e_i\otimes e_j\big]\otimes \big[e_k\otimes e_l\big]\otimes \big[e_u\otimes e_v]
\end{multline*}
for $f,g,h:[n]\to[n]$ functions such that $i\mapsto(f(i),g(i))$ is injective.
Then,
    \[\bordersubrank(T) \geq n.\]
\end{lemma}

\begin{proof}
We define the maps
\begin{align*}
&A(\eps) : [e_i \otimes e_j] \mapsto \eps^{-j} [e_i \otimes e_j]\\
&B(\eps) : [e_i \otimes e_j] \mapsto \eps^{i} [e_i \otimes e_j]\\
&C(\eps) : [e_i \otimes e_j] \mapsto [e_i \otimes e_j].
\end{align*}
Then
\begin{multline*}
(A(\eps) \otimes B(\eps) \otimes C(\eps))S = 
\sum_{i=1}^n \big[e_i\otimes e_{f(i)}\big]\otimes\big[e_{f(i)}\otimes e_{g(i)}\big] \otimes \big[e_{h(i)}\otimes e_i\big]\\
+ \sum_{i=1}^n\sum_{j<k}\sum_{l,u,v=1}^n \eps^{k-j} c_{ijkluv}\big[e_i\otimes e_j\big]\otimes \big[e_k\otimes e_l\big]\otimes \big[e_u\otimes e_v],
\end{multline*}
where the first sum is a diagonal of length~$n$.
Therefore, $\bordersubrank(T) \geq n$.
\end{proof}

\begin{proof}[Proof of \autoref{thm:samePivotsInRow_sqrtLowerBound_asymsub}]
    By \autoref{lem:upper-triangular}, $T^{\boxtimes 2}$ restricts to the upper-triangular form that is the hypothesis for \autoref{lem:border-subrank-for-upper-triag}, which in turn states that it degenerates to a diagonal tensor of size $n$. By \autoref{lem:border-asymp-subrank} this lower bounds its asymptotic subrank $\asympsubrank(T^{\boxtimes 2})\geq\bordersubrank(T^{\boxtimes 2})\geq n$. So $\asympsubrank(T)\geq\sqrt{n}$.
\end{proof}

\section{Asymptotic slice rank versus asymptotic subrank}\label{slicerank-vs-subrank}

In this section we prove a lower bound on the asymptotic subrank in terms of the asymptotic slice rank, thus bounding how different the asymptotic subrank and asymptotic slice rank can be. This requires some of the methods developed so far, and some extra ingredients. These extra ingredients are the notion of the minimal covering number of a matrix subspace, which is also known as non-commutative rank in the literature, and a Flanders-type relation to the max-rank.

\subsection{Minimal covering number of matrix subspaces and slice spans}\label{subsec:mincov}

For any matrix subspace $\mathcal{A} \subseteq \FF^{n_1} \otimes \FF^{n_2}$ we define the \emph{minimal covering number} $\mincov(\mathcal{A})$ as 
    the smallest number $m_1 + m_2$ such that there are subspaces $V_1 \subseteq \FF^{n_1}$ and $V_2 \subseteq \FF^{n_2}$ with $\dim V_1 = m_1$ and $\dim V_2 = m_2$ such that $\mathcal{A} \subseteq V_1 \otimes \FF^{n_2} + \FF^{n_1} \otimes V_2$.
In the literature, the parameter $\mincov$ is also known as the non-commutative rank~\cite{MR2123060}, and a matrix subspace $\mathcal{A}$ satisfying $\mincov(\mathcal{A}) \leq r$ is sometimes called $r$-decomposable~\cite{meshulam1985maximal}.

\begin{lemma}\label{th:Qi-SRi-matrix}
    For any matrix subspace $\mathcal{A}$, $\maxrank(\mathcal{A}) \leq \mincov(\mathcal{A})$.
\end{lemma}
\begin{proof}
    There are subspaces $V_1\subseteq\FF^{n_1}$ and $V_2\subseteq\FF^{n_2}$ with $\dim(V_1)+\dim(V_2)=\mincov(\mathcal{A})$ %
    and such that $\mathcal{A} \subseteq V_1\otimes\FF^{n_2}+\FF^{n_1}\otimes V_2$. Any matrix in $\mathcal{A}$ can then be written as $A+B$ for some $A\in V_1\otimes\FF^{n_2}$ and $B\in\FF^{n_1}\otimes V_2$. Then $\rank(A)\leq\dim{V_1}$ and $\rank(B)\leq\dim(V_2)$. This gives $\rank(A+B)\leq\rank(A)+\rank(B)\leq\mincov(\mathcal{A})$, as required.
\end{proof}

In \autoref{subsec:flanders} we discuss reverse versions of 
\autoref{th:Qi-SRi-matrix}. 

For any tensor $T \in \FF^{n_1} \otimes \FF^{n_2} \otimes \FF^{n_3}$, let $\mathcal{A}_\ell$ denote the slice spans of $T$ for $\ell \in [3]$. We define $\slicerank_\ell(T) = \mincov(\mathcal{A}_\ell)$ for $\ell \in [3]$.
\autoref{th:Qi-SRi-matrix} directly gives:

\begin{lemma}\label{th:Qi-SRi}
    For any tensor $T$ and every $\ell \in [3]$, we have $\subrank_\ell(T) \leq \slicerank_\ell(T)$.
\end{lemma}

Finally, it is easily seen how $\subrank_\ell$ is a relaxation of $\subrank$, and similarly how $\slicerank$ is a relaxation of $\slicerank_\ell$, leading to the following inequalities: %

\begin{lemma}\label{lem:ranks-vs-ranksi}
    For any tensor $T$ and every $\ell\in[3]$, we have $\subrank(T) \leq \subrank_\ell(T)$.
\end{lemma}
\begin{proof}
    Let $r=\subrank(T)$. Then $T\geq\langle r\rangle$ and specifically there are matrices $A,B,C$ such that $(A\otimes B\otimes C)T$ is a matrix of rank $r$, and a linear combination of the $\ell$-slices acted on by a matrix from the left and a matrix on the right. So the linear combination itself is of rank at least $r$.
\end{proof}

\begin{lemma}
    For any tensor $T$ and every $\ell\in[3]$, we have $\slicerank(T) \leq \slicerank_\ell(T)$.
\end{lemma}    
\begin{proof}
    %
%
    %
%
    Let $T \in \FF^{n_1} \otimes \FF^{n_2} \otimes \FF^{n_3}$.
    The slice rank $\slicerank(T)$ can be characterized \cite{sawin} as the smallest number $r$ such that there are subspaces $V_\ell \subseteq \FF^{n_\ell}$ for which $\sum_\ell \dim V_\ell = r$ and $T \in V_1 \otimes \FF^{n_2} \otimes \FF^{n_3} + \FF^{n_1} \otimes V_2 \otimes \FF^{n_3} + \FF^{n_1} \otimes \FF^{n_2} \otimes V_3$.
    Suppose $\ell = 3$. Let $s = \slicerank_3(T)$. Then $\mincov(\mathcal{A}) = s$ where $\mathcal{A}$ is the $3$-slice span of $T$. Then $\mathcal{A} \subseteq V_1 \otimes \FF^{n_2} + \FF^{n_1} \otimes V_2$ where $\dim V_1 = m_1$, $\dim V_2 = m_2$ and $m_1 + m_2 = s$. So $T \in V_1 \otimes \FF^{n_2} \otimes \FF^{n_3} + \FF^{n_1} \otimes V_2 \otimes \FF^{n_3}$, which gives $\slicerank(T) \leq s$.
\end{proof}
\subsection{Flanders-type bounds on the max-rank}\label{subsec:flanders}

We know for any matrix subspace $\mathcal{A}$ that $\maxrank(\mathcal{A}) \leq \mincov(\mathcal{A})$ %
(\autoref{th:Qi-SRi}). The goal of this section is to discuss inequalities in the reverse direction. While these are not new, we provide self-contained proofs for the convenience of the reader. Since the tensor parameters $\subrank_\ell$ and $\slicerank_\ell$ are defined in terms of $\maxrank$ and $\mincov$ we immediately get similar inequalities for those. We start with the most basic reverse inequality, which is 
due to Flanders~\cite{flanders1962spaces}:

\begin{lemma}\label{th:sr-q-2-matrix} %
    Let $\mathcal{A}$ be a matrix subspace over the field $\FF$.
    Suppose $|\FF| > \maxrank(\mathcal{A})$. Then we have that $\mincov(\mathcal{A}) \leq 2 \maxrank(\mathcal{A})$.
\end{lemma}

\autoref{th:sr-q-2-matrix} requires the base field $\FF$ to be large enough. For arbitrary fields, using a slightly larger constant on the right hand side, a similar inequality is still known to be true by a result of Haramaty and Shpilka \cite[Lemma 3.5]{haramaty2010structure} (Lemma~3.7 in the arXiv version).

\begin{lemma}\label{thm:sri_leq_4Qi-matrix} %
    Let $\mathcal{A}$ be any matrix subspace. Then $\mincov(\mathcal{A}) \leq 4 \maxrank(\mathcal{A})$.
\end{lemma}
The result in \cite[Lemma 3.5]{haramaty2010structure} is stated in the more general setting of quadratic polynomials. We follow their proof closely, but present it for the special case of matrices (bilinear forms).
\begin{proof}
    Throughout the proof we will consider matrices as vectors in the tensor product $\FF^{n_1}\otimes\FF^{n_2}$.
    Denote $r=\maxrank(\mathcal{A})$. We will find subspaces $U_1\subseteq\FF^{n_1},V_1\subseteq\FF^{n_2}$, both of dimension $r$, such that restricting $\mathcal{A}$ to some complement subspaces for them $U_1^\perp,V_1^\perp$
    we get maximum rank in $\mathcal{A}|_{U_1^\perp\times V_1^\perp}$ at most $\frac{r}{2}$ (as we are working 
    over arbitrary fields, we are not assuming an inner product. $U_1^\perp,V_1^\perp$ are complements w.r.t.\ an arbitrary direct sum decomposition of $\FF^{n_1},\FF^{n_2}$ as explained in detail below). Repeating this process, we get by induction subspaces $U_1,U_2\dots,U_m$ and $V_1,V_2,\dots,V_m$ such that restricting $\mathcal{A}$ to complements of them all we have the trivial space of only the $0$ matrix. The dimension of both $U_i$ and $V_i$ is at most $\frac{r}{2^{i-1}}$ so the dimension of both $\sum_{i=1}^m U_i$ and $\sum_{i=1}^m V_i$ is at most $\sum_{i=1}^m\frac{r}{2^{i-1}}\leq 2r$. Then $\mathcal{A}\subseteq (\sum_{i=1}^m U_i)\otimes\FF^{n_2}+\FF^{n_1}\otimes(\sum_{i=1}^m V_i)$ shows the desired $\mincov(\mathcal{A})\leq 2r+2r=4r$.

    We have left to show that indeed we can find $U_1,V_1$ as described above. Let $A\in\mathcal{A}$ be a maximal-rank element, so $\rank(A)=r$. Let $A=\sum_{i=1}^r u_i\otimes v_i$ be a rank-decomposition of $A$. Define $U_1:=\linspan\{u_1,\dots,u_r\}$ and $V_1:=\linspan\{v_1,\dots,v_r\}$. By the fact that it is a rank decomposition we have that $\{u_1,\dots, u_r\}$ and $\{v_1,\dots,v_r\}$ are both linearly independent sets of vectors, showing $\dim(U_1)=\dim(V_1)=r$. Complete them to bases $u_1,\dots,u_{n_1}$ and $v_1,\dots,v_{n_2}$ of $\FF^{n_1}$ and $\FF^{n_2}$ respectively. Define $U_1^\perp:=\linspan\{u_{r+1},\dots,u_{n_1}\}$ and $V_1^\perp:=\linspan\{v_{r+1},\dots,v_{n_2}\}$ and denote the projectors onto them with respect to these bases by $\Pi_{U_1^\perp}$ and $\Pi_{V_1^\perp}$ respectively. Define the restricted subspace $\mathcal{A}|_{U_1^\perp\times V_1^\perp}:=\{(\Pi_{U_1^\perp}\otimes\Pi_{V_1^\perp})(B):B\in\mathcal{A}\}$. 
    Let $B\in\mathcal{A}$ be any element in $\mathcal{A}$. Denote its restriction by $B|_{U_1^\perp\times V_1^\perp}$ and denote the rank of this restriction by $s:=\rank(B|_{U_1^\perp\times V_1^\perp})$. We want to show $s\leq\frac{r}{2}$. Write the restriction as a rank-decomposition $B|_{U_1^\perp\times V_1^\perp}=\sum_{i=1}^s x_i\otimes y_i$ for elements $x_i\in U_1^\perp$ and $y_i\in V_1^\perp$. Then $\{u_1,\dots,u_r,x_1,\dots,x_s\}$ and $\{v_1,\dots,v_r,y_1,\dots,y_s\}$ are both linearly independent. 
    Thus, to write the unrestriced $B$ we only need to introduce tensor products of either elements from $U_1$ with $\ell^{(2)}_i\in\FF^{n_2}$ or tensor products of $\ell^{(1)}_i\in\FF^{n_1}$ with elements from~$V_1$, that is: 
    $B=\sum_{i=1}^s x_i\otimes y_i+\sum_{i=1}^r( u_i\otimes \ell^{(2)}_i+\ell^{(1)}_i\otimes v_i)$. Decompose $\ell^{(2)}_i=\tilde{v}_i+\tilde{y}_i+\tilde{\ell}^{(2)}_i$ where $\tilde{v}_i\in V_1$, $\tilde{y}_i\in \linspan\{y_1,\dots,y_s\}$ and $\tilde{\ell}^{(2)}_i$ is independent of both subspaces (or $0$). Do the same for the $\ell^{(1)}_i$'s. Denote $\tilde{y}_i=\sum_{j=1}^s b_{i,j}y_j$ and $\tilde{x}_i=\sum_{j=1}^s a_{i,j}x_j$ for $a_{i,j},b_{i,j}\in\FF$, and define ${\ell'}^{(1)}_j:=\sum_{i=1}^r b_{i,j}u_i\in U_1$ and similarly ${\ell'}^{(2)}_j:=\sum_{i=1}^r a_{i,j}v_i\in V_1$. Then 
    \[
    B=\underbrace{\sum_{i=1}^s (x_i+{\ell'}^{(1)}_i)\otimes(y_i+{\ell'}^{(2)}_i)}_{=:C}
    \,+\, \underbrace{\sum_{i=1}^r(u_i\otimes(\tilde{v}_i+\tilde{\ell}^{(2)}_i)+(\tilde{u}_i+\tilde{\ell}^{(1)}_i)\otimes v_i)-\sum_{i=1}^s({\ell'}^{(1)}_i\otimes{\ell'}^{(2)}_i)}_{=:D}.
    \]
    Note that $\{x_i+{\ell'}^{(1)}_i\}_{i=1}^s\cup\{u_i\}_{i=1}^r\cup\{\tilde{\ell}^{(1)}_i:i\in[r],\tilde{\ell}^{(1)}_i\neq0\}$ is linearly independent and so is $\{y_i+{\ell'}^{(2)}_i\}_{i=1}^s\cup\{v_i\}_{i=1}^r\cup\{\tilde{\ell}^{(2)}_i:i\in[r],\tilde{\ell}^{(2)}_i\neq0\}$. In particular we have $\rank(C)=s$ and $B=C\oplus D$. 
    Then $s+\rank(D)=\rank(C)+\rank(D)=\rank(B)\leq r$ so $\rank(D)\leq r-s$. Recall that $A\in U_1\otimes V_1$ so $C\in \linspan\{x_i+{\ell'}^{(1)}_i\}_{i=1}^s\otimes \linspan\{y_i+{\ell'}^{(2)}_i\}_{i=1}^s$ and $D+A$ are supported on linearly independent vectors (in both factors), implying the rank is additive on them:
    $r\geq\rank(B+A)=\rank(C+D+A)=\rank(C)+\rank(D+A)\geq s+(\rank(A)-\rank(D))\geq s+(r-(r-s))=2s$ where the first inequality is by $A+B\in\mathcal{A}$, the second is generally true for any two matrices and the third is the rank of $A$ and the inequality for the rank of $D$ we established earlier.
\end{proof}

From the above follows immediately the corresponding inequalities for the parameters $\subrank_\ell$ and $\slicerank_\ell$.  %

\begin{lemma}\label{th:sr-q-2}
   Let $T \in \FF^{n_1} \otimes \FF^{n_2} \otimes \FF^{n_3}$. Suppose $|\FF| > \subrank_\ell(T)$. Then $\slicerank_\ell(T) \leq 2\subrank_\ell(T)$ for every $\ell \in [3]$.
\end{lemma}
\begin{lemma}\label{thm:sri_leq_4Qi}
    Let $T \in \FF^{n_1} \otimes \FF^{n_2} \otimes \FF^{n_3}$. Then $\slicerank_\ell(T) \leq 4 \subrank_\ell(T)$ for every $\ell \in [3]$.
\end{lemma}

\subsection{Lower bound on asymptotic subrank in terms of slice rank}\label{sec:gap}

In previous sections we have proved lower bounds on the asymptotic subrank of concise tensors that only depend on the dimensions of the tensor.
Here we shift gears, and use some of the methods that we have developed so far, including the Flanders-type bounds of \autoref{subsec:flanders}, to prove a lower bound on the asymptotic subrank in terms of the slice rank. In fact, we will prove a lower bound on the border subrank in terms of slice rank. As a consequence, we obtain a bound on the gap between asymptotic subrank and asymptotic slice rank.

\begin{lemma}\label{lem:b-subr-sr}
    Let $T \in \FF^{n_1} \otimes \FF^{n_2} \otimes \FF^{n_3}$. Then %
    $\bordersubrank(T^{\boxtimes 3})\geq
    \frac{3}{64}\slicerank(T)^2$.
\end{lemma}

\begin{proof}
    By \autoref{lem:cube-mamu} we have $T^{\boxtimes 3}\geq\langle\subrank_{1}(T),\subrank_2(T),\subrank_3(T) \rangle$.
    Strassen~\cite[Theorem~6.6]{MR882307} proved  for any $e\leq h\leq \ell$ that
    \[
    \bordersubrank(\langle e,h,\ell\rangle)\geq \begin{cases}
    eh-\lfloor\frac{(e+h-\ell)^2}{4}\rfloor &e+h\geq\ell\\
    eh &\textnormal{otherwise}.
    \end{cases}
    \]
    In both cases we have $\bordersubrank(\langle e,h,\ell\rangle)\geq \frac{3}{4}eh$ (the worst case being $\ell=h$).
    Combining these we get 
        \[\bordersubrank(T^{\boxtimes 3})\geq \tfrac34\subrank_1(T)\cdot\subrank_2(T)
        \geq\tfrac34(\tfrac14\slicerank_1(T))\cdot(\tfrac14\slicerank_2(T))\geq \tfrac34(\tfrac14\slicerank(T))^2
        \]
    for which we used \autoref{thm:sri_leq_4Qi} and \autoref{lem:ranks-vs-ranksi}.
\end{proof}

%
%
Unlike rank and subrank, slice rank is not sub- or super-multiplicative under Kronecker product, so Fekete's lemma cannot be used to guarantee its asymptotic version's existence. 
To state the following theorem, let $\asympsupslicerank(T):=\limsup_{m\to\infty} \slicerank(T^{\boxtimes m})^{1/m}$. Of course, this is just $\asympslicerank(T)$ if the limit exists.

\begin{theorem}
\label{th:q-sr-bound}
Let $T \in \FF^{n_1} \otimes \FF^{n_2} \otimes \FF^{n_3}$. Then
\[
\asympsupslicerank(T) \geq \asympsubrank(T) \geq \asympsupslicerank(T)^{2/3}.
\]
\end{theorem}
\begin{proof}
The first inequality follows directly from the fact that slice rank upper bounds subrank. For the second inequality, we have from \autoref{lem:border-asymp-subrank} and \autoref{lem:b-subr-sr} that $\asympsubrank(T) \geq \bordersubrank(T^{\boxtimes 3})^{1/3} \geq (\tfrac{3}{64}\slicerank(T)^2)^{1/3}$. Now replace $T$ by $T^{\boxtimes m}$, take the $m$th root on both sides and let $m$ go to infinity to get the claim.
\end{proof}

\begin{remark}\label{th:a-subr-sr}
 The constant in \autoref{lem:b-subr-sr} can be slightly improved if we replace the border subrank by the asymptotic subrank.
 Namely, for any $3$-tensor $T$, we have $\asympsubrank(T) \geq (\frac{1}{4} \slicerank(T))^{2/3}$.
 The proof is as follows.
 We have
    $\asympsubrank(T^{\boxtimes 3}) \geq \min_{{\ell_1},{\ell_2}} \subrank_{\ell_1}(T) \subrank_{\ell_2}(T)$ (\autoref{lem:asympsub-min}). Combining with $\subrank_\ell(T) \geq \tfrac14 \slicerank_\ell(T)$ (\autoref{thm:sri_leq_4Qi}) and $\slicerank_\ell(T) \geq \slicerank(T)$ (\autoref{lem:ranks-vs-ranksi}) gives
    \[
    \asympsubrank(T)^3 = \asympsubrank(T^{\boxtimes 3}) \geq \min_{{\ell_1},{\ell_2}} (\tfrac14\slicerank_{\ell_1}(T)) (\tfrac14\slicerank_{\ell_2}(T)) \geq (\tfrac14\slicerank(T))^2.
    \]
    Taking the cube root on both sides gives the claim. Another situation in which we can improve the constant is when the base field $\FF$ is large enough, which allows using $\subrank_\ell(T) \geq \tfrac12 \slicerank_\ell(T)$ (\autoref{th:sr-q-2}) instead of $\subrank_\ell(T) \geq \tfrac14 \slicerank_\ell(T)$ (\autoref{thm:sri_leq_4Qi}) in the above proofs, leading to a stronger bound.
\end{remark}

\section*{Acknowledgments} 
We thank the Centre de Recherches Mathématiques Montréal, where this work was initiated. 
We also thank Jurij Volcic and Vladimir Lysikov for helpful and stimulating discussions. We thank Koen Hoeberechts for useful comments. MC thanks the National Center for Competence in Research SwissMAP of the Swiss National Science Foundation and the Section of Mathematics at the University of Geneva for their hospitality. IL thanks the QMATH Center at the University of Copenhagen for the hospitality during his stay.




\bibliographystyle{alphaurl}
\bibliography{refs}


\begin{dajauthors}
\begin{authorinfo}[briet]
Jop Briët\\
Centrum Wiskunde \& Informatica\\
Amsterdam, The Netherlands\\
j\imagedot{}briet\imageat{}cwi\imagedot{}nl\\
\end{authorinfo}
\begin{authorinfo}[christandl]
Matthias Christandl\\
University of Copenhagen\\
Copenhagen, Denmark\\
christandl\imageat{}math\imagedot{}ku\imagedot{}dk
\end{authorinfo}
\begin{authorinfo}[leigh]
  Itai Leigh\\
  Tel Aviv University\\
  Tel Aviv, Israel\\
  itai\imagedot{}leigh\imageat{}mail\imagedot{}huji\imagedot{}ac\imagedot{}il
\end{authorinfo}
\begin{authorinfo}[shpilka]
  Amir Shpilka\\
  Tel Aviv University\\
  Tel Aviv, Israel\\
  shpilka\imageat{}tauex\imagedot{}tau\imagedot{}ac\imagedot{}il
\end{authorinfo}
\begin{authorinfo}[zuiddam]
  Jeroen Zuiddam\\
  University of Amsterdam\\
  Amsterdam, The Netherlands\\
  j\imagedot{}zuiddam\imageat{}uva\imagedot{}nl
\end{authorinfo}

\end{dajauthors}

\end{document}